\documentclass[12pt]{article}
\usepackage{amsmath}
\usepackage{graphicx}
\usepackage{enumerate}
\usepackage{natbib}
\usepackage{url} 
\usepackage{mathtools}
\usepackage{siunitx}
\usepackage{amsthm}
\usepackage{bm}
\usepackage{amssymb}
\usepackage{adjustbox}
\usepackage{float}
\usepackage{sectsty}
\usepackage{footnote}

\DeclarePairedDelimiter{\ceil}{\lceil}{\rceil}

\newtheorem{definition}{Definition}
\newtheorem{theorem}{Theorem}
\newtheorem{corollary}{Corollary}
\newtheorem{lemma}{Lemma}
\newtheorem{proposition}{Proposition}
\newtheorem{innercustomlemma}{Lemma}
\newenvironment{customlemma}[1]
{\renewcommand\theinnercustomlemma{#1}\innercustomlemma}
{\endinnercustomlemma}

\newcommand{\blind}{0}

\begin{document}

\def\spacingset#1{\renewcommand{\baselinestretch}%
{#1}\small\normalsize} \spacingset{1}


\if0\blind
{
  \title{\bf Space-Time Covariance Models on Networks with An Application on Streams}
  \author{Jun Tang\thanks{
    Jun Tang is Ph.D. Candidate, Department of Statistics and Actuarial Science, University of Iowa, Iowa City, IA 52242 (E-mail: jun-tang-1@uiowa.edu). Please address all correspondence to this author.}\hspace{.2cm} and Dale Zimmerman\thanks{Dale L. Zimmerman is Professor, Department of Statistics and Actuarial Science and Department of Biostatistics, University of Iowa.} \\
    Department of Statistics and Actuarial Science, University of Iowa\\}
\date{}
  \maketitle
} \fi

\if1\blind
{
  \bigskip
  \bigskip
  \bigskip
  \begin{center}
    {\LARGE\bf Space-Time Covariance Models on Networks with An Application on Streams}
\end{center}
  \medskip
} \fi

\bigskip
\begin{abstract}
The second-order, small-scale dependence structure of a stochastic process defined in the space-time domain is key to prediction (or kriging). While great efforts have been dedicated to developing models for cases in which the spatial domain is either a finite-dimensional Euclidean space or a unit sphere, counterpart developments on a generalized linear network are practically non-existent. To fill this gap, we develop a broad range of parametric, non-separable space-time covariance models on generalized linear networks and then an important subgroup --- Euclidean trees by the space embedding technique --- in concert with the generalized Gneiting class of models and 1-symmetric characteristic functions in the literature,  and the scale mixture approach. We give examples from each class of models and investigate the geometric features of these covariance functions near the origin and at infinity.   We also show the linkage between different classes of space-time covariance models on Euclidean trees. We illustrate the use of models constructed by different methodologies on a daily stream temperature data set and compare model predictive performance by cross validation.
\end{abstract}

\bigskip
\bigskip
\noindent%
\vfill

\newpage
\spacingset{1.5} 
\section{Introduction}\label{sec:1}
\subsection{Background}\label{sec:1.1}
Despite its wide variety of applications in different scientific disciplines, including environmental (see for example, ~\citet{ver2006};~\citet{cressie2006};~\citet{ver2010};  and~\citet{ODonnel2014}), neurological (\citet{jammalamadaka2013}; ~\citet{baddeley2014}), ecological~\citep{ang2012} and social sciences (\citet{ang2012};~\citet{baddeley2017}), the study of a random process observed on a network is a relatively new area in spatial statistics. Observations that are closer together in space tend to be more alike than observations far apart~\citep{tobler1970}. Thus, the small-scale (covariance) structure of a geostatistical process is usually assumed to be a function of distance between spatial locations. On a network, it is possible that two sampling sites are close together in the sense of Euclidean distance, but are far apart within the network. Under such a circumstance, it is more reasonable to use the alternative metric when modeling the dependence structure. However, merely replacing the Euclidean distance in a standard geostatistical model with the shortest path within the network may lead to an invalid (not positive definite) covariance function on the network, and thus result in negative prediction variances~\citep{ver2006}. 

In a finite-dimensional Euclidean space, the well-known Bochner's theorem~\citep[pp.~58]{bochner1955} fully characterizes the class of stationary continuous covariance functions as Fourier transforms of finite, nonnegative measures. Though this powerful theorem provides a sufficient and necessary condition for positive definiteness, closed-form Fourier inversions do not always exist. Schoenberg's result~\citep{schoenberg1938cm}, on the other hand, is Fourier transform free. It reveals the one-to-one relationship between isotropic covariance functions and completely monotone functions in a Hilbert space. Quite a few non-separable parametric spatio-temporal covariance functions have been developed based on Bochner's and Schoenberg's theorems; see, for example,~\citet{cressie1999} and \citet{gneiting2002}. \citet[pp.~526]{yaglom1987}'s kernel convolution-based approach can also be generalized into the space-time domain~\citep{rodrigues2010}.

Covariance functions receive special attention due to the fact that a Gaussian random process is completely determined by its first- and second-order moments~\citep{porcu2019}. Although a broad range of classes of space-time covariance models are available in Euclidean space~\citep{de2013} and  a thorough review has recently been given by~\citet{porcu2019}, corresponding results for pure spatial linear networks are few and far between -- a recent exception being~\citet{anderes2020} -- and space-time results on networks are practically non-existent. To fill this gap, in this article we develop Fourier-free space-time covariance functions on generalized linear networks, using space embedding and scale mixture approaches.

\subsection{Parametric Space-time Covariance Models}\label{sec:1.2}
Let $\left\{Z(\bm{s}; t): (\bm{s}, t) \in \mathcal{D} \times \mathbb{R} \right\} $ denote a univariate, real-valued, continuously-indexed stochastic process on the product space of a spatial domain $\mathcal{D}$ and a temporal domain $\mathbb{R}$. In the literature, the spatial domain is usually taken to be either a finite-dimensional Euclidean space ($\mathcal{D} = \mathbb{R}^{n}$) or a unit sphere ($\mathcal{D} = \mathbb{S}^{n}$) ~\citep{porcu2019}. In contrast, we consider a random process on a generalized linear network ($\mathcal{D} = \mathcal{G}$) whose definition will be given in the subsequent section. Assume that the first two moments of the random process exist and that the mean structure $\mu(\bm{s}; t)$, which measures the global scale space-time variability, can be fixed as a constant, i.e. $\mu(\bm{s}; t) \equiv \mu$, or modeled as a linear combination of covariates of interest, i.e. $\mu(\bm{s}; t) = \bm{x}(\bm{s}; t)'\bm{\beta}$. The second-order,  small-scale dependence structure, commonly described by a parametric covariance function, is key to space-time prediction (or kriging) and regression-type parameter estimation and is the main focus of this paper. We do not distinguish between Gaussian and non-Gaussian processes unless necessary since the covariance function plays an important role in either situation.  

To be consistent with the recent literature, we let $C$ denote the covariance function where 
\[
C(\bm{s}_{1}, \bm{s}_{2}; t_{1}, t_{2}) := \text{Cov}(Z(\bm{s}_{1}; t_{1}), Z(\bm{s}_{2}; t_{2})), \qquad (\bm{s}_{i}, t_{i}) \in \mathcal{D} \times \mathbb{R},\ i = 1, 2.
\] 
From the definition, $C$ is symmetric, i.e., $C(\bm{s}_{1}, \bm{s}_{2}; t_{1}, t_{2}) = C(\bm{s}_{2}, \bm{s}_{1}; t_{2}, t_{1})$. Moreover, a covariance function must be positive definite, meaning that
\begin{equation}\label{pd1}
	\sum_{i=1}^{N}\sum_{j=1}^{N} a_{i}a_{j}C(\bm{s}_{i}, \bm{s}_{j}; t_{i}, t_{j}) \geq 0
\end{equation} 
for any finite collections of $\{a_{i}\}_{i=1}^{N} \subset \mathbb{R}$ and $\{(\bm{s}_{i}, t_{i})\}_{i=1}^{N} \subset \mathcal{D} \times \mathbb{R}$.  Covariance functions which fail to satisfy this condition are likely to lead to negative prediction variances and undefined probability densities. Whenever a function $C$ satisfies the symmetry and positive definiteness conditions, we call it a valid covariance function. 

In geostatistics, a common assumption made by practitioners is second-order stationarity, which requires that the overall mean is constant and that the covariance function depends on the spatial locations only through their relative positions. Moreover, in Euclidean space, a covariance function is called isotropic if it is a function of the Euclidean norm of the difference between locations. Unlike its counterpart in Euclidean space, the definition of stationarity is less clear on networks~\citep{anderes2020}. Nevertheless, a space-time covariance function is said to be isotropic within components if $C(\bm{s}_{1}, \bm{s}_{2}; t_{1}, t_{2}) = f(d(\bm{s}_{1}, \bm{s}_{2}); |t_{1} - t_{2}|)$ for some function $f: [0, \infty) \times [0, \infty) \rightarrow \mathbb{R}$, where $d: \mathcal{D} \times \mathcal{D} \rightarrow [0, \infty)$ is a distance metric~\citep{anderes2020} which satisfies (1) $d(\bm{s}_{1}, \bm{s}_{2}) = d(\bm{s}_{2}, \bm{s}_{1})$, for any $\bm{s}_{1}, \bm{s}_{2} \in \mathcal{D}$; (2) $d(\bm{s}_{1}, \bm{s}_{2}) = 0$ if and only if $\bm{s}_{1} = \bm{s}_{2}$,  and $|t_{1} - t_{2}|$ is the absolute difference between times. We call such an $f$ a radial profile function. We work with either the covariance function or the radial profile function, denoting both by $C$, when the context causes no confusion. By assuming isotropy, the model guarantees that the covariance function is fully symmetric~\citep{gneiting2002} since $C(\bm{s}_{1}, \bm{s}_{2}; t_{1}, t_{2}) = C(\bm{s}_{1}, \bm{s}_{2}; t_{2}, t_{1}) = C(\bm{s}_{2}, \bm{s}_{1}; t_{1}, t_{2}) = C(\bm{s}_{2}, \bm{s}_{1}; t_{2}, t_{1})= f(d(\bm{s}_{1}, \bm{s}_{2}); |t_{1} - t_{2}|)$, for any $\bm{s}_{1}, \bm{s}_{2} \in \mathcal{D}$ and $t_{1}, t_{2} \in \mathbb{R}$. 

When it comes to spatio-temporal covariance models, assuming separability is a convenient starting point~\citep{rodrigues2010}. Specifically, a space-time model is separable if it can be written as a product or a sum of pure spatial and pure temporal models, i.e.,
\[
C(\bm{s}_{1}, \bm{s}_{2}; t_{1}, t_{2}) = C_{S}(\bm{s}_{1}, \bm{s}_{2}) \times C_{T}(t_{1}, t_{2}) \quad \text{or} \quad C(\bm{s}_{1}, \bm{s}_{2}; t_{1}, t_{2}) = C_{S}(\bm{s}_{1}, \bm{s}_{2}) + C_{T}(t_{1}, t_{2}),
\] 
for all space-time coordinates $(\bm{s}_{1}, t_{1})$, $(\bm{s}_{2}, t_{2}) \in \mathcal{D} \times \mathbb{R}$. Given that $C_{S}$ and $C_{T}$ are valid covariance functions on $\mathcal{D}$ and $\mathbb{R}$, the product or the sum is valid on $\mathcal{D} \times \mathbb{R}$. It has been argued that the class of separable covariance models is severely limited due to the lack of space-time interaction~\citep{de2013}, and that in many cases it implies ``unphysical dependence among process variables"~\citep{porcu2019}. We limit our attention to non-separable covariance functions in this paper. 

\subsection{Overview and Contributions}\label{sec:1.3}

In this paper, we adopt both the space embedding technique and the scale mixture approach to construct a broad range of valid space-time covariance functions on generalized linear networks and/or Euclidean trees. The rest of the paper is organized as follows. Section~\ref{sec:2} reviews preliminaries about generalized linear networks equipped with two distance metrics: resistance distance and geodesic distance. Section~\ref{sec:3} gives sufficient conditions for constructing isotropic space-time covariance models by space embedding on arbitrary generalized linear networks and then on an important subgroup -- Euclidean trees. Besides deriving space-time covariance functions on directed Euclidean trees based on the scale mixture approach and the convex cone property in Section~\ref{sec:4}, we also show that the exponential tail-down model~\citep{ver2010} is the one and only that is directionless (i.e. isotropic) and is thereby a bridge between models in Section~\ref{sec:3} and Section~\ref{sec:4}. In Section~\ref{sec:5}, we apply a few space-time covariance models on daily stream temperature measurements from a stream network in the northwest United States and compare model predictive performance. Section~\ref{sec:6} concludes the paper with discussion.   

\section{Generalized Linear Networks and Distance Metrics}\label{sec:2}
\subsection{Generalized Linear Networks}\label{sec:2.1}
A network, also called a graph, is a collection of vertices (nodes) joined by edges~\citep[chap.~1]{mn} and is denoted by the pair $(\mathcal{V}, \mathcal{E})$. A linear network is the union of finitely many line segments in the plane where different edges only possibly intersect with each other at one of their vertices (see Figure~\ref{fig:eg}). It is useful to associate each edge with a positive real number, which is called the weight. Weights can be physical edge lengths, strengths, etc. The space-time covariance functions in our paper are defined on generalized linear networks, whose definition was introduced by~\citet{anderes2020} and is revisited below.

\begin{definition}\label{def1}
	A triple $\mathcal{G} = (\mathcal{V}, \mathcal{E}, \{\xi_{e}\}_{e \in \mathcal{E}})$ which satisfies conditions (I) - (IV) is called a graph with Euclidean edges. 
	
	\noindent (I) Graph structure: $(\mathcal{V}, \mathcal{E})$ is a finite simple connected graph, meaning that the vertex set $\mathcal{V}$ is finite, the graph has no self-edges or multi-edges, and every pair of vertices is connected by a path.
	
	\noindent (II) Edge sets: Each edge $e\in \mathcal{E}$ is associated with a unique abstract set, also denoted by $e$. The vertex set $\mathcal{V}$ and  all the edge sets are mutually disjoint. 
	
	\noindent (III) Edge coordinates: For each edge $e \in \mathcal{E}$ and vertices $u, v \in \mathcal{V}$ joined by $e$, there is a bijective mapping $\xi_{e}$ defined on the union of the edge set $e$ and vertices $\{u, v\}$, i.e. $e \cup \{u, v\}$, such that $\xi_{e}$ maps $e$ onto an open interval $(\underline{e}, \overline{e}) \subset \mathbb{R}$ and $\{u, v\}$ onto endpoints $\{\underline{e}, \overline{e}\}$, respectively. 
	
	\noindent (IV) Distance consistency: Define $d_{\mathcal{G}}(u, v): \mathcal{V} \times \mathcal{V} \rightarrow [0, \infty)$ as the length of the shortest path on vertices of a weighted graph where the weight associated with each edge $e \in \mathcal{E}$ is defined as $\overline{e} - \underline{e}$. Then, the following equality holds: 
	\[
	d_{\mathcal{G}}(u, v) = \overline{e} - \underline{e}
	\]
	for every $e \in \mathcal{E}$ connecting two vertices $u, v \in \mathcal{V}$. 
\end{definition}

In our work, we assume that the topological structure of $\mathcal{G}$ does not evolve over time. Any arbitrary site $\bm{s}$ on such a network $\mathcal{G}$ is denoted as $\bm{s} \in \mathcal{G} = \bm{s} \in \mathcal{V} \cup \bigcup_{e\in\mathcal{E}}$. Graphs with Euclidean edges extend the notion of traditional linear networks by including graphs which do not have a planar representation in $\mathbb{R}^{2}$ (see, e.g., Figure 3 of~\citet{anderes2020}). We use the terms generalized linear networks and 
graphs with Euclidean edges interchangeably in this paper. Any tree-like graph $(\mathcal{V}, \mathcal{E})$ is planar and can be constructed as a graph with Euclidean edges easily. We call such a graph a Euclidean tree and denote it by $\mathcal{T}$. Vertices of a Euclidean tree connected with only one edge are called leaves. 

\begin{figure}[H]
	\begin{center}
		\includegraphics[scale=1]{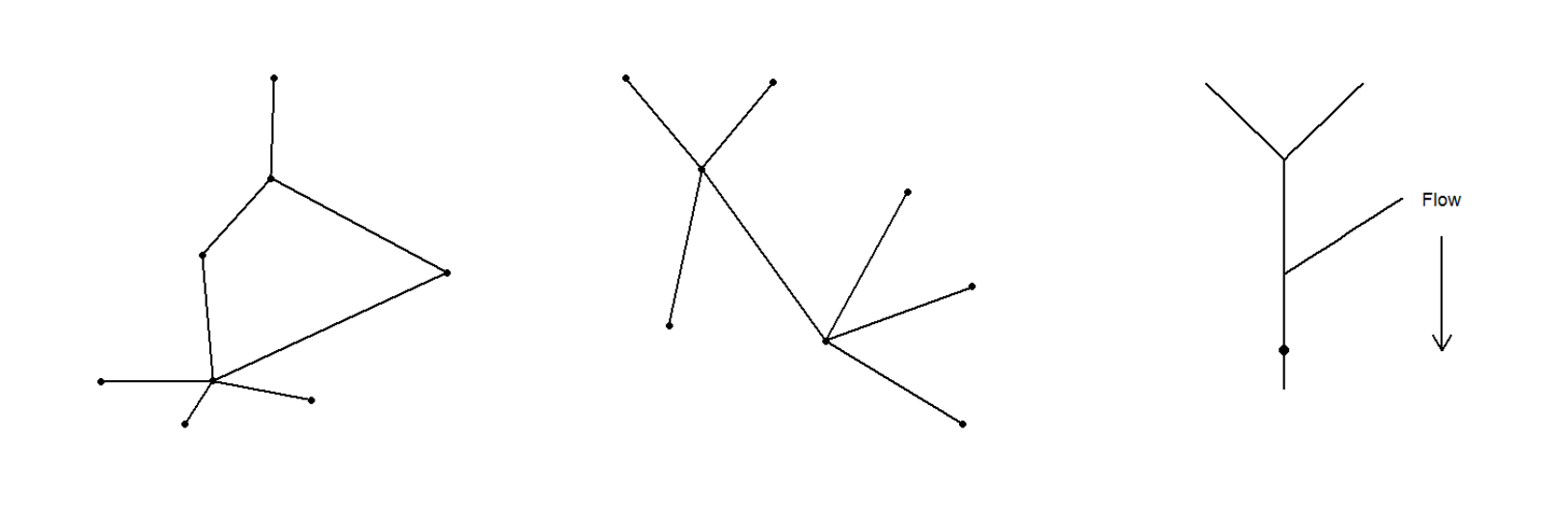}
	\end{center}
	\caption{Examples of networks. The middle panel is a directionless tree, while the right most panel is a directed tree (e.g. stream) with the outlet superimposed.}\label{fig:eg}
\end{figure}

\subsection{Distance Metrics}\label{sec:2.2}
Let $\Phi(\mathcal{D}, d)$ denote the class of radial profile functions such that for any $f \in \Phi(\mathcal{D}, d)$, 
\begin{equation}\label{pd2}
	\sum_{i=1}^{N}\sum_{j = 1}^{N}a_{i}a_{j} f(d(\bm{s}_{i}, \bm{s}_{j})) \geq 0,
\end{equation}
for any finite collection $\{a_{i}\}_{i=1}^{N}\subset \mathbb{R}$ and $\{\bm{s}_{i}\}_{i=1}^{N} \subset \mathcal{D}$. For any $f \in \Phi(\mathcal{D}, d)$, we say that $f$ is positive definite on $\mathcal{D}$ with respect to distance $d$. When the space domain is a finite-dimensional Euclidean space $\mathcal{D} = \mathbb{R}^{n}$, for any $ \bm{x}, \bm{y} \in \mathbb{R}^{n}$, where $\bm{x} = (x_{1}, \cdots, x_{n})'$ and $\bm{y} = (y_{1}, \cdots, y_{n})'$, let $d_{p}$ be the standard $l_{p}$ norm with
\[
d_{p}(\bm{x}, \bm{y}) = ||\bm{x} - \bm{y}||_{p} =  \left(\sum_{i=1}^{n}\left|x_{i} - y_{i}\right|^{p}\right)^{1/p}, \qquad 1 \leq p < \infty.  
\]
When the space domain is a real Hilbert space $\mathcal{H}$, the norm, denoted by $||\cdot||_{\mathcal{H}}$ is induced by the inner product $<\cdot, \cdot>$, such that
\[
||\bm{x}||_{\mathcal{H}} := \sqrt{<\bm{x}, \bm{x}>}, \qquad \bm{x} \in \mathcal{H}.
\]

A generalized linear network $\mathcal{G}$ comes along with two distance metrics: one is the standard length of shortest path, a.k.a. geodesic distance or stream distance, denoted by $d_{G, \mathcal{G}}$; the other is resistance distance $d_{R, \mathcal{G}}$. $d_{R, \mathcal{G}}$ is defined as the variogram of an auxiliary random field $Y_{\mathcal{G}}$: 
\[
d_{R, \mathcal{G}}(\bm{s}_{1}, \bm{s}_{2}) := \text{Var} \left(Y_{\mathcal{G}}(\bm{s}_{1}) - Y_{\mathcal{G}}(\bm{s}_{2}) \right), \qquad \forall \bm{s}_{1}, \bm{s}_{2} \in \mathcal{G},
\]
whose formal construction is given by~\citet{anderes2020} and we skip the details. The resistance metric is an extension of the one in electrical network theory from pairs of vertices to any points on the graph. Both metrics satisfy the two conditions mentioned in the previous section and have the relationship given in Proposition~\ref{prop1}, which is a portion of Proposition 4 of~\citet{anderes2020}:

\begin{proposition}\label{prop1}
	For a graph with Euclidean edges $\mathcal{G}$,
	\begin{equation*}
		d_{R, \mathcal{G}} (\bm{s}_{1}, \bm{s}_{2}) \leq d_{G, \mathcal{G}} (\bm{s}_{1}, \bm{s}_{2}), \qquad \bm{s}_{1}, \bm{s}_{2} \in \mathcal{G}.
	\end{equation*} 
	Equality holds if and only if $\mathcal{G}$ is a Euclidean tree. 
\end{proposition}

Since, by Proposition~\ref{prop1}, $d_{R, \mathcal{T}} = d_{G, \mathcal{T}}$, henceforth we let $d_{\cdot, \mathcal{T}}$ denote either one.

\section{Isotropic Space-time Models by Space Embedding}\label{sec:3}
In this section, we adopt the space embedding technique to transform the abstract, less familiar spatial domain, $\mathcal{G}$, to simpler, well-studied spaces, e.g. Hilbert and Euclidean spaces, and build isotropic space-time models from there. 

We follow the definition of isometric spaces given by~\citet{anderes2020} and ~\citet{menegatto2020}. Define a distance space as a pair $(\mathcal{D}, d)$, where $\mathcal{D}$ is a non-empty set and the function $d$ is specified in Section~\ref{sec:1.2}. 

\begin{definition}\label{def2}
	A distance space $(\mathcal{D}, d)$ is said to be  $g-$embeddable in a Hilbert space $(\mathcal{H}, ||\cdot||_{\mathcal{H}})$ if $g: [0, \infty) \rightarrow [0, \infty)$, and there exists a mapping $i: \mathcal{D} \rightarrow \mathcal{H}$ such that 
	\[
	g(d(\bm{s}_{1}, \bm{s}_{2})) = ||i(\bm{s}_{1}) - i(\bm{s}_{2})||_{\mathcal{H}}, \qquad \bm{s}_{1}, \bm{s}_{2} \in \mathcal{D}.
	\]
	If function $g$ is the identity map, we say $(\mathcal{D}, d)$ is isometrically embeddable in $\mathcal{H}$. 
\end{definition}

\subsection{Hilbert Space Embedding of Generalized Linear Network}\label{sec:3.1}
Our first main contribution is based on the square-root embedding result of a graph with Euclidean edges into a Hilbert space proved by~\citet{anderes2020} and restated below. 

\begin{theorem}\label{thm1}
	\textbf{(square-root embedding, Anderes et al.)} Given $\mathcal{G}$ a graph with Euclidean edges, there exists a Hilbert space $\mathcal{H}$ and a mapping $i: \mathcal{G} \rightarrow \mathcal{H}$ such that
	\[
	\sqrt{d_{R, \mathcal{G}}(\bm{s}_{1}, \bm{s}_{2})} = ||i(\bm{s}_{1}) - i(\bm{s}_{2})||_{\mathcal{H}} 
	\]
	for all $\bm{s}_{1}, \bm{s}_{2} \in \mathcal{G}$. In the special case in which, $\mathcal{G}$ is a Euclidean tree, the above result also holds for the geodesic distance.
\end{theorem}

The so-called Gneiting class of covariance functions has been especially popular in space-time geostatistical modeling~\citep{porcu2019} and will be the building block of the isotropic covariance functions given in this section. Despite some discrepancy in the literature, here a function $\varphi: [0, \infty) \rightarrow \mathbb{R}$ is said to be completely monotone on $[0, \infty)$ if $\varphi$ is continuous on $[0, \infty)$, infinitely differentiable on $(0, \infty)$ and $(-1)^{j}\varphi^{(j)}(t) \geq 0$ over $(0, \infty)$ for every integer $j \geq 0$, where $\varphi^{(j)}$ denotes the $j$th derivative of $\varphi$ and $\varphi^{(0)} = \varphi$. A nonnegative continuous function $\psi(t): [0,\infty) \rightarrow \mathbb{R}$ with completely monotone derivative is called a Bernstein function. In analogy to the definition of positive definite functions, we recall that for a distance space $(\mathcal{D}, d)$, a continuous function $f: \mathcal{D} \rightarrow \mathbb{R}$ is called \emph{conditionally negative definite} (see for example in~\citet{menegatto2020}) on $\mathcal{D}$ with respect to $d$ if
\begin{equation}\label{cnd}
	\sum_{i=1}^{N}\sum_{j=1}^{N}a_{i}a_{j}f(d(\bm{s}_{i}, \bm{s}_{j})) \leq 0,
\end{equation}
for any finite collections of $\{a_{i}\}_{i=1}^{N} \subset \mathbb{R}$ and $\{\bm{s}_{i}\}_{i=1}^{N} \subset \mathcal{D}$, given $\sum_{i=1}^{N} a_{i} = 0$. For such a function we write $f \in CND(\mathcal{D}, d)$.

We are now ready to formulate and prove our first main result. Denote the generalized \emph{Gneiting} class of continuous functions by $G_{\alpha}$, where 
\[
G_{\alpha}(d, u) = \frac{1}{\psi(d)^{\alpha}}\varphi\left(\frac{u}{\psi(d)} \right), \quad d, u \geq 0, 
\]
with $\psi$ and $\varphi$ being strictly positive and continuous. Theorem 3.2 in~\citet{menegatto2020} provides sufficient conditions for $G_{\alpha}$ to be positive definite over the product space of a quasi-metric space and a finite Euclidean space, which extends~\citet{gneiting2002}'s results in Euclidean spaces and will be applied on $\mathcal{G} \times \mathbb{R}$.

\begin{theorem}\label{thm2}
	\textbf{(generalized Gneiting class)} Let $G_{\alpha}$ be the function defined above with $\varphi(\cdot)$ being completely monotone. Assume that $a \in (0, 1]$ and $\alpha \geq 1/2$. Then for any pairs $(\bm{s}_{1}, t_{1}), (\bm{s}_{2}, t_{2}) \in \mathcal{G} \times \mathbb R$, where $\mathcal{G}$ is equipped with the resistance distance $d_{R, \mathcal{G}}$, the following statements are true:
	
	\noindent \textbf{(a)} the function $C(d_{R, \mathcal{G}}(\bm{s}_{1}, \bm{s}_{2}); |t_{1}- t_{2}|):= G_{\alpha}(d_{R, \mathcal{G}}(\bm{s}_{1}, \bm{s}_{2}), |t_{1} - t_{2}|^{2a})$ is a valid covariance function over $\mathcal{G}\times \mathbb{R}$ provided that $\psi \in CND(\mathcal{G}, d_{R, \mathcal{G}})$;
	
	\noindent \textbf{(b)} the function $C(d_{R, \mathcal{G}}(\bm{s}_{1}, \bm{s}_{2}); |t_{1}- t_{2}|):= G_{\alpha}(d_{R, \mathcal{G}}(\bm{s}_{1}, \bm{s}_{2}), |t_{1} - t_{2}|^{2a})$ is a valid covariance function over $\mathcal{G}\times \mathbb{R}$ provided that $\psi := g \circ h$, where $g$ is a positive Bernstein function and $h$ is a nonnegative valued function such that $h \in CND(\mathcal{G}, d_{R, \mathcal{G}})$;
	
	\noindent \textbf{(c)} the function $C(d_{R, \mathcal{G}}(\bm{s}_{1}, \bm{s}_{2}); |t_{1}- t_{2}|):= G_{\alpha}(d_{R, \mathcal{G}}(\bm{s}_{1}, \bm{s}_{2})^{b}, |t_{1} - t_{2}|^{2a})$ is a valid covariance function over $\mathcal{G}\times \mathbb{R}$ provided that $b \in (0, 1]$ and $\psi$ is a positive Bernstein function.  
	
	\noindent Moreover, when $\mathcal{G}$ is a Euclidean tree, the above results hold for $d_{G, \mathcal{G}}$ as well.
\end{theorem}

\begin{proof}
	The symmetry of the functions defined in all three parts is obvious. A distance space $(\mathcal{D}, d)$ defined in this paper is also a quasi-metric space in~\citet{menegatto2020}, while the converse is not necessarily true. Therefore, parts (a) and (b) of Theorem~\ref{thm2} are direct applications of~\citet{menegatto2020}'s work to the case where the dimension of the Euclidean space is $1$ and the quasi-metric space is $(\mathcal{G}, d_{R, \mathcal{G}})$. By Theorem~\ref{thm1}, $(\mathcal{G}, d_{R, \mathcal{G}})$ is square root-embeddable in a Hilbert space $\mathcal{H}$. Notice that $(\mathcal{G}, \sqrt{d_{R, \mathcal{G}}})$ is again a distance space and thus isometrically embeddable into $\mathcal{H}$ by Definition~\ref{def2}. Statement (c) follows in concert with Theorem 3.2 (iii) in~\citet{menegatto2020}. When $\mathcal{G}$ is a Euclidean tree, Proposition~\ref{prop1} gives that $d_{R, \mathcal{G}} = d_{G, \mathcal{G}}$, which completes the proof. 
\end{proof}

Each part (a) - (c) of Theorem~\ref{thm2} provides researchers an easy to implement method for constructing valid space-time covariance functions on a generalized linear network. Let us digress for a moment and consider a pure spatial, isotropic random process $Z(\bm{s})$ defined on $(\mathcal{G}, d_{R, \mathcal{G}})$ with an overall constant mean $\mu$. Similar to the conclusion in geostatistics, the semivariogram, defined as $
\gamma(d_{R, \mathcal{G}}(\bm{s}_{1}, \bm{s}_{2})) := \frac{1}{2}Var(Z(\bm{s}_{1}) - Z(\bm{s}_{2}))$
is conditionally negative definite, i.e. $\gamma \in CND(\mathcal{G}, d_{R, \mathcal{G}})$. This result holds for $d_{G, \mathcal{G}}$ as well when $\mathcal{G}$ is Euclidean tree. Following the previous notation, let $Cov(Z(\bm{s}_{1}), Z(\bm{s}_{2})) = C(\bm{s}_{1}, \bm{s}_{2}) = f(d_{R, \mathcal{G}}(\bm{s}_{1}, \bm{s}_{2}))$. Then by definition, there exists the following relationship
\begin{equation}\label{cnd2}
	\gamma(d_{R, \mathcal{G}}(\bm{s}_{1},\bm{s}_{2})) = f(0) - f(d_{R, \mathcal{G}}(\bm{s}_{1}, \bm{s}_{2})),
\end{equation} 
for all $\bm{s}_{1}, \bm{s}_{2} \in \mathcal{G}$, between $\gamma$ and $f$. Hence, given any radial profile $f \in \Phi(\mathcal{G}, d_{R, \mathcal{G}})$, we can construct $\gamma$ based on (\ref{cnd2}), which belongs to $CND(\mathcal{D}, d_{R, \mathcal{G}})$. For examples of the class $\Phi(\mathcal{G}, d_{R, \mathcal{G}})$, we refer to~\citet{anderes2020}. Note that statements (a) - (c) are not exclusive. For instance, $\psi(t) = t^{\lambda} + \beta$, for $t \geq 0$ with $0 < \lambda \leq 1, \beta > 0$ is a positive Bernstein function (given in Table~\ref{tb:fnc}) and also belongs to $CND(\mathcal{T}, d_{\cdot, \mathcal{T}})$ (see Lemma 3 in Supplement A). When $b = 1$ in (c), both (a) and (c) give the same subclass of valid covariance functions on $\mathcal{T} \times \mathbb{R}$.  

In addition to constructing valid covariance functions on $\mathcal{G} \times \mathbb{R}$, we also investigate the geometric features of marginal functions whose definition will be given shortly, near the origin and at infinity. It has been discussed in~\citet{de2010} that spatial and temporal marginals, defined as $f_{S}(d) := f(d, 0)$ and $f_{T}(u) := f(0, u)$, respectively (where $f$ denotes the space-time   radial profile function), play a significant role in selecting an appropriate and physically meaningful class of covariances in applications. By comparing empirical covariance functions with estimated ones, any obvious disagreement would indicate model misspecification~\citep{stein2005}.  

It is clear that the covariance function, as well as both marginal functions, constructed by Theorem~\ref{thm2} are continuous at the origin since $\psi$ and $\varphi$ are continuous on $[0, \infty)$. Although L\'evy-Khinchin's formula~\citep[pp.~15-45]{berg2008} characterizes the class of conditionally negative definite functions in $\mathbb{R}^{n}$, analogous results in the distance space $(\mathcal{D}, d)$, especially $(\mathcal{G}, d_{R, \mathcal{G}})$, have not been obtained, as far as we know. Hence, we defer marginal results related to $CND(\mathcal{G}, d_{R, \mathcal{G}})$ for future research and investigate properties of marginals pertaining to the covariance functions by Theorem~\ref{thm2}(c) only.

\begin{proposition}\label{prop2}
	Let $C(d_{R, \mathcal{G}}(\bm{s}_{1}, \bm{s}_{1}); |t_{1}- t_{2}|):= G_{\alpha}(d_{R, \mathcal{G}}\left(\bm{s}_{1}, \bm{s}_{2})^{b}, |t_{1} - t_{2}|^{2a} \right)$, where $\alpha \geq 1/2,\ a \in (0, 1],\  b \in (0, 1]$, $\varphi$ is completely monotone and $\psi$ is a positive Bernstein function. Then the spatial $f_{S}$ and temporal $f_{T}$ marginal functions, i.e., $f_{S}(d) = G_{\alpha}(d^{b}, 0)$ and $f_{T}(u) = G_{\alpha}\left(0, u^{2a} \right)$, are non-increasing over $[0, \infty)$. 
\end{proposition}

\begin{proof}
	Observe that $f_{S}(d) = \varphi(0)\psi \left(d^{b} \right)^{-\alpha}$. The first-order derivative of the spatial marginal gives
	\[
	f'(d) = -\alpha b \varphi(0) \left(d^{b-1} \right)\psi\left(d^{b} \right)^{-\alpha -1}\psi'\left(d^{b}\right)  \leq 0, \quad d \in (0, \infty),
	\]
	since $\varphi$ and $\psi$ are positive functions and the derivative $\psi'$ is completely monotonic, thus nonnegative. Similarly, one can show the first-order derivative of the temporal marginal is
	\[
	f'_{T}(u) = \frac{2a}{\psi(0)^{\alpha + 1}} u^{2a - 1}\varphi'\left(\frac{u^{2a}}{\psi(0)} \right) \leq 0, \quad u \in (0, \infty),
	\]
	since $\varphi' \leq 0$ on $(0, \infty)$.
\end{proof}

\begin{proposition}\label{prop3}
	Let $C(d_{R, \mathcal{G}}(\bm{s}_{1}, \bm{s}_{1}); |t_{1}- t_{2}|):= G_{\alpha}\left(d_{R, \mathcal{G}}(\bm{s}_{1}, \bm{s}_{2})^{b}, |t_{1} - t_{2}|^{2a} \right)$, where $\alpha \geq 1/2, \ a \in (0, 1], \ b \in (0, 1]$, $\varphi$ is completely monotone and $\psi$ is a positive Bernstein function. Then the spatial marginal function $f_{S}(d) = G_{\alpha}\left(d^{b}, 0\right)$, is convex on $(0, \infty)$. 
\end{proposition}

\begin{proof}
	By direct calculation, the second-order derivative of the spatial marginal is
	\begin{align*}
		f''_{S}(d) &= -(\alpha b \varphi(0))\left[  -(\alpha + 1) b \left(d^{b-1} \right)^{2} \psi\left(d^{b} \right)^{-\alpha - 2}\left(\psi'\left(d^{b} \right) \right)^{2} + \right. \\
		& \quad \quad \quad \qquad \qquad \left. b \left(d^{b-1} \right)^{2}\psi'\left(d^{b} \right)^{-\alpha - 1} \psi''\left(d^{b} \right) + (b-1)\left(d^{b-1} \right)^{2}
		\psi'\left(d^{b} \right) ^{-\alpha - 1}\psi'\left(d^{b} \right) \right],
	\end{align*}
	where each term in the brackets is non-positive for $d \in (0, \infty)$.
\end{proof}

Note that the temporal marginal function $f_{T}$ does not share the convexity property in general. Justified by Proposition~\ref{prop2}, space-time covariance functions constructed by Theorem~\ref{thm2}(c) satisfy the physical law~\citep{tobler1970} which says observations that are closer in space and time have higher correlation. The asymptotic behavior of the model, e.g. $\lim_{d \rightarrow \infty}G_{\alpha}\left(d^{b}, u^{2a} \right)$ for fixed $u$ and $\lim_{u \rightarrow \infty}G_{\alpha}\left(d^{b}, u^{2a} \right)$ for fixed $d$, depends on the asymptotic behavior of $\psi$ and $\varphi$, respectively, thus does not present a unified conclusion, in general.

We list a few examples of completely monotone functions and positive Bernstein functions in Table~\ref{tb:fnc}, which can be found in~\citet{gneiting2002}, ~\citet{mill2001},  and~\citet[pp.~15-45]{berg2008}.

\begin{table}[ht]
	\caption{Examples of Completely Monotone Functions $\varphi(t)$ and Positive Bernstein Functions $\psi(t)$, $t \geq 0$.}\label{tb:fnc}
	\begin{adjustbox}{width=\textwidth}	
		\begin{tabular}{ll|ll}
			\hline
			\hline
			Function & Parameters & Function & Parameters\\
			\hline
			$\varphi(t) = \exp(-c t^{\nu})$ & $c > 0, 0 < \nu \leq 1$ & 
			$\psi(t) = (\kappa t^{\lambda} + 1)^{\beta}$ & $\kappa > 0, 0 \leq \beta \leq 1, 0 < \lambda \leq 1$\\
			$\varphi(t) = \exp\left(ct^{\nu} \right)$ & $c > 0, \nu < 0$ 
			& $\psi(t) = \frac{\log(\kappa t^{\lambda} + \beta)}{\log(\beta)}$ & $\kappa > 0, \beta > 1, 0 < \lambda \leq 1$ \\
			$\varphi(t) = \left(\frac{2}{\exp(ct^{1/2}) + \exp(-ct^{1/2})} \right)^{\nu}$ & $c > 0, \nu > 0$  & $\psi(t) = t^{\lambda} + \beta$ & $0 < \lambda \leq 1, \beta > 0$  \\
			$\varphi(t) = \left(1 + ct^{\gamma} \right)^{-\nu}$ &  $c > 0, \nu > 0, 0 < \gamma \leq 1$  & $\psi(t) = \beta - \exp(-\kappa t)$ & $\kappa > 0, \beta > 1$ \\
			\hline
		\end{tabular}
	\end{adjustbox}
\end{table}

More completely monotone functions can be built by constructive tools, such as the property that the class of functions is closed under addition and multiplication~\citep{mill2001}.  We show the application of Theorem~\ref{thm2} by a couple of examples.
\vspace{4mm}

\noindent \emph{Example 1.} Consider the completely monotone function and the Bernstein function from the first row of Table~\ref{tb:fnc}. To avoid the model being overly complicated, assume $\alpha = \frac{1}{2}$, $a = 1$ and write the geodesic distance between sites as $d$ and the time lag as $u$ throughout. The space-time covariance function $C_{0}$ given by Theorem~\ref{thm2}(c) may be written as follows:
\[
C_{0}(d; u) = \frac{1}{\left(\kappa d^{b \lambda} + 1  \right)^{\beta /2}}\exp\left(-c\left(\frac{u^{2}}{\left(\kappa d ^{b \lambda} + 1 \right)^{\beta}} \right)^{\nu} \right),
\]
where $c > 0, 0 < \nu \leq 1, \kappa > 0, 0 \leq \beta \leq 1, 0 < \lambda \leq 1$ and $0 < b\leq 1$. Since $b$ and $\lambda$ appear only as the product $b\lambda$, and both have the same support, i.e. $(0, 1]$, henceforth to avoid an identification issue we use $b$ instead of $b\lambda$. 

Imitating Example 1 from~\citet{gneiting2002}, consider the pure spatial covariance function $C_{S}$~\citep{anderes2020} on $(\mathcal{G}, d_{R, \mathcal{G}})$:
\[
C_{S}(d) = (\kappa d^{b} + 1)^{-\delta}, 
\]
where $\kappa > 0, 0 < b \leq 1$ and $\delta \geq 0$. By the Schur product theorem~\citep{schur1911}, it follows that the product of $C_{0}$ and $C_{S}$ also defines a valid space-time covariance function on $\mathcal{G} \times \mathbb{R}$. After reparameterization (i.e. let $\tau = \frac{\beta}{2} + \delta$), we have
\begin{equation}\label{M1}
	C(d; u) = \frac{1}{\left(\kappa d^{b} + 1  \right)^{\tau}}\exp\left(-c\left(\frac{u^{2}}{\left(\kappa d ^{b} + 1 \right)^{\beta}} \right)^{\nu} \right),  
\end{equation}
where $c > 0, 0 < \nu \leq 1, \kappa > 0, 0 \leq \beta \leq 1, \tau \geq \frac{\beta}{2}$ and $0 < b \leq 1$. We call $\beta$ the space-time interaction parameter since when $\beta = 0$, the covariance function becomes separable. The spatial and temporal marginals, as well as the covariance function itself, all decay to zero as $d \rightarrow \infty$ and/or $u \rightarrow \infty$, which indicates the same variability in the spatial and temporal components in the sense discussed by~\citet{de2013} (visualizations of marginal functions, along with the covariance surface can be found in Suppplement B). We will come back to the model given by (\ref{M1}) later. 

\vspace{4mm}

\noindent \emph{Example 2.} Let $\varphi(t) = \left(\frac{2}{\exp(ct^{1/2}) + \exp(-ct^{1/2})} \right)^{\nu}$ and $\psi(t) = t^{\lambda} + 1$, where $t \geq 0$ with parameters $c > 0, \nu > 0$, and $0 < \lambda \leq 1$. Then Theorem~\ref{thm2}(c) gives another valid space-time covariance function:
\begin{equation}\label{M2}
	C(d; u) = \frac{2^{\nu}}{\left(d^{b\lambda} + 1 \right)^{\alpha}} \left\{\exp \left(c\frac{u^{a}}{\left(d^{b\lambda} + 1 \right)^{1/2}}\right) + 
	\exp \left( - c\frac{u^{a}}{\left(d^{b\lambda} + 1 \right)^{1/2}} \right) \right\}^{-\nu}, 
\end{equation}
where $0 < a \leq 1, \alpha \geq 1/2, 0 < b \leq 1, c > 0, \nu > 0$, and $ 0 < \lambda \leq 1$. Again, $\lambda$ only appears as a multiplier of $b$ and both have the same support, i.e. $(0, 1]$, so we will drop $\lambda$.  The asymptotic behaviors of (\ref{M2}), and its corresponding marginals,  are the same as in \emph{Example 1}.

\subsection{$l_{1}$ Embedding of Euclidean Trees}\label{sec:3.2}
Now, we focus on a subclass of generalized linear networks $\mathcal{G}$: the Euclidean trees $\mathcal{T}$. In addition to the square-root embedding result which holds for any $\mathcal{G}$,~\citet{anderes2020} provides another space embedding result for $\mathcal{T}$ only, which is restated below.

\begin{theorem}\label{thm3}
	\textbf{($l_{1}$ embedding, Anderes et al.)} Let $\mathcal{T}$ be a Euclidean tree with $m$ leaves, where $m \geq 3$. Then $(\mathcal{T}, d_{\cdot, \mathcal{T}})$ is isometrically embeddable into $(\mathbb{R}^{n}, \rho_{1})$, where $n = \ceil[\big]{\frac{m}{2}}$ and $\rho_{1}$ is the $l_{1}$ norm,  such that there exists a mapping $i: \mathcal{T} \rightarrow \mathbb{R}^{n}$ satisfying:
	\[
	d_{\cdot, \mathcal{T}}(\bm{s}_{1}, \bm{s}_{2}) = \rho_{1}(i(\bm{s}_{1})- i(\bm{s}_{2})), 
	\]  
	for any $\bm{s}_{1}, \bm{s}_{2} \in \mathcal{T}$. 
\end{theorem}

The $l_{1}$ embedding result comes directly from the proof pertaining to Theorem 4 in~\citet{anderes2020}. Meanwhile, the positive definite functions on $(\mathbb{R}^{n}, d_{1})$ are essentially the same as $\alpha-$symmetric (here $\alpha = 1$) characteristic functions in $\mathbb{R}^{n}$ by Bochner's theorem, which have been extensively studied by~\citet{cambanis1983}, ~\citet{gneiting1998}, and~\citet{zastavnyi2000} and will play a fundamental role in constructing space-time covariance functions on $\mathcal{T} \times \mathbb{R}$. Before we dive into our second main contribution, we give the definition of \emph{linear} isometric embedding due to~\citet{zastavnyi2000}. The term \emph{linear} is added to distinguish from Definition~\ref{def2}. 

\begin{definition}\label{def3}
	Suppose that $\mathcal{L}_{i}$ is a linear space, and a function $\rho_{i}: \mathcal{L}_{i} \rightarrow [0, \infty)$ exists, that satisfies $\rho_{i}(c\bm{x}) = |c|\rho_{i}(\bm{x})$, for any scalar $c$ and $\bm{x} \in \mathcal{L}_{i}$, $i = 1, 2$. The pair $(\mathcal{L}_{1}, \rho_{1})$ is said to be \emph{linearly} isometrically embeddable in $(\mathcal{L}_{2}, \rho_{2})$ if there is a linear operator $A: \mathcal{L}_{1} \rightarrow \mathcal{L}_{2}$ such that $\rho_{1}(\bm{x}) = \rho_{2}(A\bm{x})$ for all $\bm{x} \in \mathcal{L}_{1}$. If either of $(\mathcal{L}_{1}, \rho_{1})$ and $(\mathcal{L}_{2}, \rho_{2})$ is linearly isometrically embeddable in the other, we call  these spaces \emph{linearly} isometric. 
\end{definition}

Being slightly different from (\ref{pd2}), let $\Phi(\mathbb{R}^{n}, \rho)$ denote the class of functions such that for any $f \in \Phi(\mathbb{R}^{n}, \rho)$,
\[
\sum_{i=1}^{N}\sum_{j=1}^{N}a_{i}a_{j}f(\rho(\bm{x}_{i} - \bm{x}_{j})) \geq 0,
\]
for any finite collection $\{a_{i}\}_{i=1}^{N} \subset \mathbb{R}$ and $\{\bm{x}_{i}\}_{i=1}^{N} \subset \mathbb{R}^{n}$. For any such $f$ and $\rho$, we say that $f \circ \rho$ is positive definite on $\mathbb{R}^{n}$. 

\begin{lemma}\label{lemma1}
	Consider $(\mathbb{R}^{n}, \rho_{1})$ and $(\mathbb{R}^{n}, \rho_{2})$, where $\rho_{1}$ is the $l_{1}$ norm and $\rho_{2}(\bm{x}) = \sum_{i =1}^{n} \frac{1}{c_{i}}|x_{i}|$ with  $\{c_{i}\}_{i=1}^{n}$ being fixed positive scalars, for $\bm{x} \in \mathbb{R}^{n}$. Then we have $\Phi(\mathbb{R}^{n}, \rho_{1}) = \Phi(\mathbb{R}^{n}, \rho_{2})$.
\end{lemma}

\begin{proof} 
	Let $A$ be the diagonal matrix with elements $c_{1}, \cdots, c_{n}$ on the main diagonal. Then for all $\bm{x} \in \mathbb{R}^{n}$, we have $\rho_{1}(\bm{x}) = \sum_{i =1}^{n}|x_{i}|  = \sum_{i =1}^{n}\frac{1}{c_{i}}|c_{i}x_{i}| = \rho_{2}(A\bm{x})$. By Definition~\ref{def3}, $(\mathbb{R}^{n}, \rho_{1})$ and $(\mathbb{R}^{n}, \rho_{2})$ are linearly isometric. Lemma~\ref{lemma1} then follows from Lemma 2(2) in~\citet{zastavnyi2000}.
\end{proof}

Let $\bm{\theta}$ denote the vector of covariance parameters and $\Theta_{n}$ the parameter space with subscript $n$ emphasizing that the dependence relates to $\mathbb{R}^{n}$. The theorem below gives a general framework for constructing valid space-time covariance functions on $\mathcal{T} \times \mathbb{R}$ by $l_{1}$ embedding.

\begin{theorem}\label{thm4}
	\textbf{(metric models)} Suppose that $\mathcal{T}$ is a Euclidean tree with $m$ leaves, where $m \geq 3$. Define $n = \ceil[\big]{m/2}$. If $f_{\bm{\theta}}\in \Phi(\mathbb{R}^{n+1}, \rho_{1})$, where $\bm{\theta} \in \Theta_{n + 1}$,  then $C(d_{\cdot, \mathcal{T}}; u) := f_{\bm{\theta}}\left(\frac{d_{\cdot, \mathcal{T}}}{\alpha} + \frac{u}{ \beta}\right)$, where $\alpha, \beta > 0$ and $\bm{\theta} \in \Theta_{n+1}$, is a valid covariance function on $\mathcal{T} \times \mathbb{R}$. 
\end{theorem}

\begin{proof}
	Let $\alpha = c_{1} = \cdots = c_{n}$ and $\beta = c_{n+1}$ in $\rho_{2}$, from Lemma~\ref{lemma1}. It follows that if $f_{\bm{\theta}}(\rho_{1}(\bm{x})) = f_{\bm{\theta}}(\sum_{i=1}^{n+1}|x_{i}|)$ is positive definite on $\mathbb{R}^{n+1}$ for $\bm{\theta} \in \Theta_{n+1}$, then $f_{\bm{\theta}}\circ \rho_{2} = f_{\bm{\theta}}\left(\frac{\sum_{i=1}^{n}|x_{i}|}{\alpha} + \frac{|x_{n+1}|}{\beta} \right)$ is also positive definite on $\mathbb{R}^{n+1}$ given $\alpha, \beta > 0$,  in addition to $\bm{\theta} \in \Theta_{n+1}$ . In concert with Theorem~\ref{thm3}, there exists a mapping $i: \mathcal{T} \rightarrow \mathbb{R}^{n}$ such that 
	\begin{align*}
		\sum_{i = 1}^{N} \sum_{j = 1}^{N} a_{i}a_{j} C(d_{\cdot, \mathcal{T}}(\bm{s}_{i}, \bm{s}_{j}); |t_{i} - t_{j}|) &=  \sum_{i = 1}^{N} \sum_{j = 1}^{N} a_{i}a_{j} f_{\bm{\theta}}\left(\frac{d_{\cdot, \mathcal{T}}(\bm{s}_{i}, \bm{s}_{j})}{\alpha} + \frac{|t_{i} - t_{j}|}{ \beta}\right) \\
		&= \sum_{i = 1}^{N} \sum_{j = 1}^{N} a_{i}a_{j} f_{\bm{\theta}}\left(
		\frac{\rho_{1}(i(\bm{s}_{i}) - i(\bm{s}_{j}))}{\alpha} + \frac{|t_{i} - t_{j}|}{ \beta}
		\right) \\
		&= \sum_{i = 1}^{N} \sum_{j = 1}^{N} a_{i}a_{j} f_{\bm{\theta}}\left(
		\frac{\sum_{k=1}^{n} |i(\bm{s}_{i})_{k} - i(\bm{s}_{j})_{k}|}{\alpha} + \frac{|t_{i} - t_{j}|}{ \beta}
		\right) \\
		& \geq 0,
	\end{align*} 
	for any finite collection $\left\{a_{i}\right\}_{i=1}^{N} \subset \mathbb{R}$ and points $\left\{\bm{s}_{i}\right\}_{i=1}^{N} \subset \mathcal{T}$, where $\alpha, \beta > 0$ and $\bm{\theta} \in \Theta_{n+1}$. 
\end{proof}

The scaling parameters, $\alpha$ and $\beta$, make the spatial and temporal distances comparable. In the literature, $\frac{d_{\cdot, \mathcal{T}}}{\alpha} + \frac{u}{ \beta}$ has been called the space-time distance and the models given by Theorem~\ref{thm4} have been called metric models~\citep{de2013}. Combining Theorem~\ref{thm4} with the sufficient conditions of 1-symmetric characteristic functions given by~\citet{cambanis1983} and~\citet{gneiting1998}, we have the following corollaries.

\begin{corollary}\label{corollary1}
	Suppose that $\mathcal{T}$ is a Euclidean tree with $m$ leaves, where $m \geq 3$. The function $C(d_{\cdot, \mathcal{T}}; u): = C_{0}\left(\frac{d_{\cdot, \mathcal{T}}}{\alpha} + \frac{u}{\beta}\right)$, where $C_{0}: [0, \infty) \rightarrow \mathbb{R}$,  is a valid covariance function on $\mathcal{T} \times \mathbb{R}$, provided that any of the following conditions is true:
	
	\noindent \textbf{(a)} $C_{0}$ is continuous, bounded, and
	\[
	\int_{0}^{\infty} C_{0}(st)f_{\ceil[\big]{m/2} + 1}(t)dt, \qquad \qquad s \geq 0,
	\]
	is the Laplace transformation of a nonnegative random variable, where $f_{n}$ is the function 
	\[
	f_{n}(r) = \frac{2 \Gamma(n)}{\left(\Gamma(n/2) \right)^{2}} r^{n-1} (1 + r^{2})^{-n}, \quad r \geq 0.
	\]
	
	\noindent \textbf{(b)}  $C_{0}$ is given by
	\[
	C_{0}(s) = \int_{0}^{\infty}\omega_{\ceil[\big]{m/2} + 1}(st)d\mu(t),
	\]
	where $\mu$ is a finite positive measure on $[0, \infty)$ and $\omega_{n}$ is defined by
	\[
	\omega_{n}(t) = \frac{\Gamma(n/2)}{\sqrt{\pi} \Gamma((n-1)/2)} \int_{1}^{\infty} \Omega_{n}(\nu^{1/2}t)\nu^{-n/2}(\nu -1)^{(n-3)/2}d\nu 
	\]
	with $\Omega_{n}(t) = \Gamma(n/2)(2/t)^{(n-2)/2}J_{(n-2)/2}(t)$ and $J_{\nu}(t)$ denoting the Bessel function of the first kind with order $\nu$. 
	
	\noindent \textbf{(c)} $C_{0}$ is a continuous function such that $C_{0}(0) = 1$,  $C_{0}^{2\ceil[\big]{m/2}}$ is convex and $\lim_{t \rightarrow \infty}C_{0}(t) = 0$.  
\end{corollary}

\begin{proof}
	Let $n = \ceil[\big]{m/2}$. Given condition (a), Theorem 3.3 in~\citet{cambanis1983} shows that $C_{0} \circ \rho_{1}$ is positive definite on $\mathbb{R}^{n+1}$. It follows immediately from Theorem~\ref{thm4} that the space-time covariance function $C_{0}\left(\frac{d_{\cdot, \mathcal{T}}}{\alpha} + \frac{u}{\beta}\right)$ is positive definite on $\mathcal{T} \times \mathbb{R}$. Parts (b) and (c)  may be proved similarly using Theorem 3.1 in~\citet{cambanis1983} and Theorem 3.2 in~\citet{gneiting1998}. 
\end{proof}

Pure spatial results, such as Theorems 4 and 5 in~\citet{anderes2020}, are applications of the $l_{1}$ embedding technique to Theorem 3.1 in~\citet{cambanis1983} and Theorem 3.2 in~\citet{gneiting1998}, respectively. 

Although Corollary~\ref{corollary1} provides sufficient conditions for space-time covariance functions to be positive definite, some of them are hard to check. We thus end this section with a corollary giving an explicit example, by applying Theorem~\ref{thm4} to a  1-symmetric characteristic function. 

\vspace{4mm}

\noindent \emph{Example 3. (powered linear with sill models)} The parametric model given by the following corollary belongs to the family of metric models and has the powered linear with sill representation.

\begin{corollary}\label{corollary2}
	Suppose that $\mathcal{T}$ is a Euclidean tree with $m$ leaves, where $m \geq 3$. Then $C(d_{\cdot, \mathcal{T}}; u): = \left(1 - \left(\frac{d_{\cdot, \mathcal{T}}}{\alpha} + \frac{u}{\beta}\right)^{\nu}\right)_{+}^{\delta}$, where $\delta \geq 2\ceil[\big]{\frac{m}{2}} + 1$, $\alpha, \beta > 0$ and $\nu \in (0, 1]$, is a valid covariance function on $\mathcal{T} \times \mathbb{R}$.
\end{corollary}

\begin{proof}
	Let $f(x):= (1 - x^{\nu})_{+}^{\delta}$, where $x \geq 0$. It is clear from~\citet{zastavnyi2000} that $f\circ \rho_{1}$ is positive definite on $\mathbb{R}^{n}$, given $\delta \geq 2n - 1$ and $\nu \in (0, 1]$, for any positive integer $n$. Therefore, $f\left(\frac{d_{\cdot, \mathcal{T}}}{\alpha} + \frac{u}{\beta} \right) = \left(1 - \left(\frac{d_{\cdot, \mathcal{T}}}{\alpha} + \frac{u}{\beta}\right)^{\nu}\right)_{+}^{\delta}$ with $\delta \geq 2\ceil[\big]{\frac{m}{2}} + 1, \nu \in (0, 1]$ and $\alpha, \beta > 0$ is positive definite on $\mathcal{T} \times \mathbb{R}$ by Theorem~\ref{thm4}. 
\end{proof}

The powered linear with sill model given by Corollary~\ref{corollary2} is continuous near the origin. Both marginals, as well as the covariance function itself, monotonically decay to zero. By direct calculation, one can show that both spatial and temporal marginal functions are convex near the origin. One feature of this model is that it has compact support, i.e. it reaches zero when the space-time distance is sufficiently large. This property is appealing for modeling large scale space-time datasets. However, its parameter space depends on the topology of the graph, i.e., the number of leaves $m$. For a fixed space-time distance, the more leaves in the tree, the smaller the dependence between observations. 

\section{Space-time Models on Directed Euclidean Trees}\label{sec:4}
Instead of working with heavily mathematically involved functions, the kernel convolution-based (or moving average) models tackle the problem from another perspective. According to \citet[pp.~526]{yaglom1987}, a large class of stationary covariance functions on the real line can be obtained by constructing a random process $\left\{Z(x): x \in \mathbb{R} \right\}$, which convolves a square-integrable kernel function $g(\cdot)$ over a white noise process $Y(x)$ defined on $\mathbb{R}^{1}$ as:
\[
Z(x) = \int_{-\infty}^{\infty} g(s - x) d Y(s), \qquad x, s \in \mathbb{R}.  
\]
When $Y(x)$ is Brownian motion, the induced covariance function is valid and can be shown to be
\[
Cov(Z(x), Z(x + h)) = C(x, x+h) = \int_{-\infty}^{\infty} g(x)g(x - h)dx, \qquad x, h \in \mathbb{R}. 
\]  
The kernel convolution-based approach allows considerable flexibility and can be generalized to nonstationary~\citep{fuentes2002}, space-time~\citep{rodrigues2010} and tree-like network~\citep{ver2010} settings. Details of the latter generalization are given below. 

\subsection{Tail-up and Tail-down Models}\label{sec:4.1}
The space-time covariance functions given in the previous section are isotropic, which might not always be an appropriate assumption due to the fact that some networks are directed in nature. For instance, in streams, flow direction is yet another important factor, in addition to shortest path length (geodesic distance), that researchers should take into consideration when modeling physical processes. Variables that move passively downstream,  e.g. chemical particles, and variables that may move upstream, e.g. fish and insects~\citep{ver2010} may need to be modeled differently. Especially for the former, we may want to allow the correlation between locations that do not share flow to be small or even zero. Based on the kernel convolution approach, \citet{ver2006} and \citet{ver2010} introduce the unilateral tail-up and tail-down covariance models on streams, which manage to handle these two scenarios differently. For detailed discussion of the models, we refer readers to~\citet{ver2010}. Here, we only give the most necessary background, which will later become essential components in our space-time covariance functions on tree-like networks.

The dendritic structure of streams guarantees that condition (I) in Definition~\ref{def1} holds. Since every tree-like network is planar, we follow the prescription of~\citet{anderes2020} by letting each edge set $e \in \mathcal{E}$ be the interior of the corresponding line segment in $\mathbb{R}^{2}$ and letting $\mathcal{V}$ be the set of endpoints of the line segments.  Moreover, let the bijection $\xi_{e}$ preserve the path-length parameterization of each line segment. Thus, conditions (I) - (IV) in Definition~\ref{def1} are satisfied and a stream equipped with stream distance, denoted by $(\mathcal{T}, d_{\cdot, \mathcal{T}})$, is a (directed) Euclidean tree. Note that models built in this section can apply to any directed Euclidean tree, which we call a stream for convenience. 

Depending on the flow direction, the tail-up and tail-down models also assume there exists a single most downstream location, which is called the outlet (see the right panel of Figure~\ref{fig:eg}). Let the index set of all stream segments be denoted by $A$, and let the index set of stream segments that are upstream of site $\bm{s}_{i} \in \mathcal{T}$, including the segment where $\bm{s}_{i}$ resides, be denoted by $U_{i} \subseteq A$.  We say two sites $\bm{s}_{i}$ and $\bm{s}_{j}$ are ``flow-connected" if they share water, i.e. if $U_{i} \cap U_{j} \neq \emptyset$, and are ``flow-unconnected" if the water at one location does not flow to the other, i.e. if $U_{i} \cap U_{j} = \emptyset$.  Equivalently, two sites are called flow-connected if and only if one is on the path of the other downstream to the outlet. The pure spatial tail-up and tail-down models are given below, where the unilateral kernel $g(x)$ is nonzero only when $x > 0$. 
\begin{itemize}
	\item Tail-up models: 
	\[
	C_{TU}(\bm{s}_{1}, \bm{s}_{2}) = \begin{cases}
	\pi_{1, 2}\int_{d}^{\infty}g(x)g(x - d) dx & \text{if }  \bm{s}_{1}, \bm{s}_{2} \text{ are flow-connected} \\
	0 & \text{if }  \bm{s}_{1}, \bm{s}_{2} \text{ are flow-unconnected},
	\end{cases} 
	\]
	where $d$ is the stream distance between sites $\bm{s}_{1}$ and $\bm{s}_{2}$ (i.e. $d_{\cdot, \mathcal{T}}(\bm{s}_{1}, \bm{s}_{2})$) and $\pi_{1, 2}$ is a weight defined as follows. Let $\Omega(x)$ be a positive additive function such that $\Omega(x)$ is constant within a stream segment, but is the sum of each segment's value when two segments join at a junction, following the flow direction. Then the weight $\pi_{1, 2} = \sqrt{\frac{\Omega(\bm{s}_{1})}{\Omega(\bm{s}_{2})}} \wedge \sqrt{\frac{\Omega(\bm{s}_{2})}{\Omega(\bm{s}_{1})}}$ ensures a constant variance of the process. In the literature, there exist different weighting schemes, see~\citet{cressie2006} and~\citet{ver2006}, and they have been proven equivalent~\citep{ver2010}. 
	
	\item Tail-down models:
	\begin{align*}
		C_{TD}(\bm{s}_{1}, \bm{s}_{2}) = \begin{cases}
			\int_{-\infty}^{-d} g(-x)g(-x - d) dx & \text{if }  \bm{s}_{1}, \bm{s}_{2} \text{ are flow-connected} \\
			\int_{-\infty}^{-a\vee b}g(-x)g(-x-|b-a|) dx & \text{if }  \bm{s}_{1}, \bm{s}_{2} \text{ are flow-unconnected}, 
		\end{cases}
	\end{align*}
	where $d$ has the same definition as in tail-up models and $a$, $b$ represent the distances from each site to the nearest junction downstream of which it shares flow with the other site. When $\bm{s}_{1}, \bm{s}_{2} \in \mathcal{T}$ are flow-unconnected, $d = a + b$.
\end{itemize}

Commonly used kernels on streams can be found in~\citet{ver2010}. Obviously, a nontrivial tail-up covariance function cannot be isotropic as the covariance is always zero when sites are flow-unconnected, while a tail-down model is a function of $a$ and $b$, in general. It has been shown by~\citet{ver2006} and~\citet{ver2010} that when the kernel is exponential, i.e. $g(x) = \theta_{1}\exp(-x/\theta_{2})$ for $x \geq 0$, $\theta_{1}, \theta_{2} > 0$, the tail-down model is a function of the geodesic distance $d_{\cdot, \mathcal{T}}$ alone, regardless of flow-connectedness. Before introducing our next main contribution in terms of space-time covariance functions, we prove that the exponential kernel is the one and only which makes the tail-down model depend on $d_{\cdot, \mathcal{T}}$ alone, or in other words, isotropic. 

\begin{theorem}\label{thm5}
	A tail-down model is isotropic, such that there exists a function $f_{TD}$, $C_{TD}(\bm{s}_{1}, \bm{s}_{2}) = f_{TD}(d_{\cdot, \mathcal{T}}(\bm{s}_{1}, \bm{s}_{2}))$ for any $\bm{s}_{1}, \bm{s}_{2} \in \mathcal{T}$, if and only if the kernel is exponential. 
\end{theorem}

The proof of Theorem~\ref{thm5} is nontrivial and left to Supplement A. When the kernel is exponential, the isotropic tail-down model can be written as
\begin{equation}\label{exp}
	C_{TD}(\bm{s}_{1}, \bm{s}_{2}) = \theta_{0} \exp(-d_{\cdot, \mathcal{T}}(\bm{s}_{1}, \bm{s}_{2})/\theta_{2}), \qquad \bm{s}_{1}, \bm{s}_{2} \in \mathcal{T},
\end{equation}
where $\theta_{0}, \theta_{2} > 0$. (\ref{exp}) also appears in~\citet{anderes2020}, where all isotropic covariance functions are developed by space embedding, as a valid covariance function on $(\mathcal{G}, d_{R, \mathcal{G}})$. Therefore, Theorem~\ref{thm5} shows that the exponential tail-down covariance function is the only bridge which connects pure spatial covariance functions on Euclidean trees constructed by space embedding and kernel convolution, and will later help us find the linkage of space-time models constructed by different approaches as well. 

\subsection{Convex Cone and Scale Mixture Models}\label{sec:4.2}
\subsubsection{Convex Cone}
Stemming from the convex cone property of the class of positive definite functions, Theorem~\ref{thm6} provides easy to implement, yet practically important, ways to construct space-time covariance functions on directed Euclidean trees. 

\begin{theorem}\label{thm6}
	The functions given below are valid space-time covariance functions on a directed Euclidean tree:
	\begin{align}
		C(\bm{s}_{1}, \bm{s}_{2}; t_{1}, t_{2}) &= C_{TD}(\bm{s}_{1}, \bm{s}_{2})C_{T1}(t_{1}, t_{2}) + C_{TU}(\bm{s}_{1}, \bm{s}_{2}) + C_{T2}(t_{1}, t_{2}) \\
		C(\bm{s}_{1}, \bm{s}_{2}; t_{1}, t_{2}) &= C_{TU}(\bm{s}_{1}, \bm{s}_{2})C_{T1}(t_{1}, t_{2}) + C_{TD}(\bm{s}_{1}, \bm{s}_{2}) + C_{T2}(t_{1}, t_{2}) \\
		C(\bm{s}_{1}, \bm{s}_{2}; t_{1}, t_{2}) &= C_{TU}(\bm{s}_{1}, \bm{s}_{2})C_{T1}(t_{1}, t_{2}) + C_{TD}(\bm{s}_{1}, \bm{s}_{2})C_{T2}(t_{1}, t_{2}),
	\end{align}
	where the tail-up $C_{TU}$ and the tail-down $C_{TD}$ models are defined in Section~\ref{sec:4.1}, and $C_{T1}$ and $C_{T2}$ are valid temporal covariance functions. 
\end{theorem}

\begin{proof}
	The symmetry condition holds trivially as each component on the right hand side of (8) - (10) is symmetric. According the Schur product theorem~\citep{schur1911}, $C_{TD}(\cdot, \cdot)C_{T1}(\cdot, \cdot)$, $C_{TU}(\cdot, \cdot)C_{T1}(\cdot, \cdot)$, and $C_{TD}(\cdot, \cdot)C_{T2}(\cdot, \cdot)$, are positive definite on $\mathcal{T} \times \mathbb{R}$. The remaining results then follow easily from the definition of positive definiteness. 
\end{proof}

In Euclidean space, (8) and (9) are called product-sum models~\citep{DeIaco2002}. Unless $C_{T1} = C_{T2}$ in (10), covariance functions given in Theorem~\ref{thm6} are non-separable. Similar to the variance components model in~\citet{ver2010}, these functions allow high autocorrelation among sites that are flow-connected, and small but significant autocorrelation among sites that are flow-unconnected, at fixed temporal components. If we further assume that the number of observations over time on each site is the same, then substantial computational efficiency can be gained by exploiting the covariance matrix structure.

\vspace{4mm}

\noindent \emph{Example 4.} Consider the isotropic exponential tail-down model, $C_{T1}$ being a cosine function which captures potential seasonal fluctuations, and $C_{T2}$ exponential as well. For the tail-up spatial component, we adopt a Mariah kernel~\citep{ver2010}, which specifies $g(x) = \frac{1}{2} \frac{1}{1 + x/\theta_{1}}$, for $x \geq 0$ with $\theta_{1}> 0$. After reparameterization, expression (10) from Theorem~\ref{thm6} gives the valid space-time model as
\[
C(\bm{s}_{1}, \bm{s}_{2}; t_{1}, t_{2}) = \begin{cases}
\frac{\pi_{1, 2}}{2} \frac{\log \left(d/\theta_{1} + 1 \right)}{d / \theta_{1}}\cos \left(\frac{u}{\theta_{2}} \right) + \frac{1}{2}\exp \left(- \left(\frac{d}{\theta_{3}} + \frac{u}{\theta_{4}} \right) \right) & \bm{s}_{1}, \bm{s}_{2} \text{ are flow-connected, } d > 0 \\
\frac{1}{2}\cos \left(\frac{u}{\theta_{2}} \right) + \frac{1}{2}\exp \left(- \frac{u}{\theta_{4}} \right) & d = 0\\
\frac{1}{2}\exp \left(- \left(\frac{d}{\theta_{3}} + \frac{u}{\theta_{4}} \right) \right) & \text{otherwise} 
\end{cases},
\]
where $\theta_{1}, \cdots, \theta_{4} > 0$ and the weight $\pi_{1, 2}$ is defined in Section~\ref{sec:4.1}. Note that the model above contains a metric sub-model which is a function of the space-time distance, i.e. $\frac{d}{\theta_{3}} + \frac{u}{\theta_{4}}$.

\subsubsection{Scale Mixture Models}
We conclude the theoretical development of space-time models with a clever trick~\citep{porcu2019}, which gives the so-called scale mixture model. The trick can trace back to the second stability property of covariance functions given by~\citet[pp.~60]{chiles1999}.    

\begin{theorem}\label{thm7}
	Let $C_{0}(\bm{s}_{1}, \bm{s}_{2}; t_{1}, t_{2}; a)$ be a space-time covariance model on $\mathcal{G} \times \mathbb{R}$, $a$ be a parameter where $a \in \Theta_{a} \subset \mathbb{R}$, and $\mu(\cdot)$ a positive measure on the set $\Theta_{a}$. Then 
	\[
	C(\bm{s}_{1}, \bm{s}_{2}; t_{1}, t_{2}) = \int_{\Theta_{a}}C_{0}(\bm{s}_{1}, \bm{s}_{2}; t_{1}, t_{2}; a) d \mu(a),
	\]
	$(\bm{s}_{i}, t_{i}) \in \mathcal{G} \times \mathbb{R}, i = 1, 2$, is a valid covariance model on $\mathcal{G} \times \mathbb{R}$ given that the integral exists for every pair of space-time coordinates. 
\end{theorem}

Theorem~\ref{thm7} can be proved by the definition of positive definiteness directly, which we will skip here. Any valid space-time model that satisfies the condition can be chosen as the integrand, and we emphasize its application on directed Euclidean trees in the following corollary.   

\begin{corollary}\label{corollary3}
	Suppose that $C_{S}(\cdot, \cdot)$  is a pure spatial covariance function on a directed Euclidean tree $\mathcal{T}$ and $C_{T}(\cdot, \cdot)$ is a pure temporal covariance function. Parameter $a$ has the support $\Theta_{a} \subset \mathbb{R}$, and $\mu(\cdot)$ is a positive measure on the set $\Theta_{a}$. Then
	\[
	C(\bm{s}_{1}, \bm{s}_{2}; t_{1}, t_{2}) = \int_{\Theta_{a}}C_{S}(\bm{s}_{1}, \bm{s_{2}}; a)C_{T}(t_{1}, t_{2}; a) d \mu(a), \qquad (\bm{s}_{i}, t_{i}) \in \mathcal{T} \times \mathbb{R}, i = 1, 2,
	\]
	is a valid space-time covariance function on $\mathcal{T} \times \mathbb{R}$ given that the integral exists. 
\end{corollary}

The proof of Corollary~\ref{corollary3} follows from Theorem~\ref{thm7} in concert with the fact that the separable space-time function $C_{S}(\bm{s}_{1}, \bm{s_{2}}; a)C_{T}(t_{1}, t_{2}; a)$ is a valid covariance function on $\mathcal{T} \times \mathbb{R}$. We illustrate the use of Corollary~\ref{corollary3} by two examples, one of which shows an interesting linkage between space-embedding models and scale mixture models due to Theorem~\ref{thm5}. 

\vspace{4mm}

\noindent \emph{Example 1 Revisit.} Following Example 4 in~\citet{DeIaco2002}, consider an exponential tail-down model as the spatial component ($C_{S}$), a cosine function as the temporal component ($C_{T}$) and a half-normal probability density function ($\mu'$), which are parameterized as follows:
\begin{align*}
	C_{S}(d; a) &= \exp\left(-\frac{a^{2}}{\theta_{1}}d \right) \\
	C_{T}(u; a) &= \cos\left[a(2\theta_{2} u) \right] \\
	\mu'(a) &= \frac{2}{\sqrt{\pi}} \exp(-a^{2}), 
\end{align*}
where $a > 0, \theta_{1} > 0, \theta_{2} \in \mathbb{R}$.  According to Corollary~\ref{corollary3},
\begin{equation}~\label{eg1}
	C(d; u) = \int_{0}^{\infty} \exp\left(-\frac{a^{2}}{\theta_{1}}d \right) \cos\left[a(2\theta_{2} u) \right]  \frac{2}{\sqrt{\pi}} \exp(-a^{2}) da,
\end{equation}
is a valid space-time covariance model on directed Euclidean tree. In order to integrate (11), we use the result from~\citet{ng1969table} that
\begin{equation}~\label{eg2}
	\int_{0}^{\infty} \exp(-x^{2})\cos (cx) dx = \frac{\sqrt{\pi}}{2}\exp\left(-\frac{c^{2}}{4}
	\right), \qquad c \in \mathbb{R}. 
\end{equation}
Now we can simplify~(\ref{eg1}) as
\begin{align}\label{eg3}
	C(d; u) &=  \int_{0}^{\infty} \exp\left(-\frac{a^{2}}{\theta_{1}}d \right) \cos\left[a(2\theta_{2} u) \right]  \frac{2}{\sqrt{\pi}} \exp(-a^{2}) da \nonumber \\
	&= \frac{2}{\sqrt{\pi}} \int_{0}^{\infty} \cos\left[a(2\theta_{2} u) \right] \exp\left(-\left(1 + \frac{d}{\theta_{1}} \right)a^{2} \right)da \nonumber \\
	&= \frac{2}{\sqrt{\pi}\sqrt{1 + \frac{d}{\theta_{1}}}}\int_{0}^{\infty} \cos \left(\frac{2 \theta_{2} u}{\sqrt{1 + \frac{d}{\theta_{1}}}}y \right) \exp\left(-y^{2} \right)dy \nonumber \\
	&= \frac{1}{\sqrt{1 + \frac{d}{\theta_{1}}}} \exp\left(-\frac{\theta_{2}^{2}u^{2}}{1 + \frac{d}{\theta_{1}}} \right), \qquad  \theta_{1} > 0, \theta_{2} \in \mathbb{R}, 
\end{align}
where the second to last equality holds by change of variable and the last by~(\ref{eg2}). Observe that the model given by~(\ref{eg3}) is a special case of~(\ref{M1}) in \emph{Example 1}, where $\kappa = 1/\theta_{1}$, $c = \theta_{2}^{2}$, $\nu = b = \beta = 1$, and $\tau = 1/2$. 

\begin{lemma}\label{lemma2}
	Assume that the scale mixture space-time covariance function is isotropic, that $C_{S}$ is an exponential tail-down model and $C_{T}$ depends on the time lag $u$ only. Then $C_{T}(0; a) \geq 0$ for $a \in \Theta_{a}$ is a sufficient but not necessary condition for the spatial marginal function $f_{S}(d):= C(d; 0)$ to be convex on $\mathbb{R}^{+}$. 
\end{lemma}

\begin{proof}
	Let $C(d; u) := \int_{\Theta_{a}} \exp \left(-\frac{a^{2}}{c}d \right)C_{T}(u; a)d \mu(a)$. Suppose that the measure $\mu(\cdot)$ and the temporal covariance function $C_{T}$ are smooth enough to allow interchanging the order of differentiation and integration. Then
	\[
	f_{S}{''}(d) = \int_{0}^{\infty} \frac{a^{4}}{c^{2}} \exp(-\frac{a^{2}}{c}d)C_{T}(0; a)d\mu(a) \geq 0, 
	\]
	given that $C_{T}(0; a) \geq 0 $ for all $a \in \Theta_{a}$. However, this condition is not necessary since we have shown in \emph{Example 1 Revisit} that when $C_{T}(u; a) = cos\left[a(2\theta_{2} u) \right]$, the scale mixture covariance model is essentially a special case of \emph{Example 1} and by Proposition~\ref{prop3}, the spatial marginal is always convex on $(0, \infty)$.
\end{proof}

Apart from the sufficient condition provided in Lemma~\ref{lemma2}, convex cone and scale mixture models do not share unified geometric features.  

\vspace{4mm}

\noindent \emph{Example 5.}  Again, let the spatial component be an exponential tail-down model with a slightly different parameterization: $C_{S}(d; a) = \exp \left(-\frac{a}{\theta_{1}}d \right)$, where $a, \theta_{1} > 0$. Then consider a non-degenerate temporal covariance function $C_{T}(u;a)= \exp \left(- \frac{a}{\theta_{2}} u^{\theta_{3}} \right)$, where $a, \theta_{2} > 0$ and $\theta_{3} \in (0, 2]$~\citep{zastavnyi2011}. Let $\mu$ be a Gamma distribution whose density function is specified as $f(a) = \frac{\theta_{5}^{\theta_{4}}}{\Gamma(\theta_{4})}a^{\theta_{4} -1}e^{-\theta_{5}a}$ for $a \in (0, \infty)$, with $\theta_{4}, \theta_{5} > 0$. The space-time covariance model based on Corollary~\ref{corollary3} is isotropic and may be evaluated as follows:
\begin{align*}
	C(d; u) &= \int_{0}^{\infty} \exp\left(-\frac{d}{\theta_{1}} a\right)\exp \left(-\frac{u^{\theta_{3}}}{\theta_{2}}a \right)\frac{\theta_{5}^{\theta_{4}}}{\Gamma(\theta_{4})}a^{\theta_{4} -1}\exp(-\theta_{5}a) d a \\
	&= \frac{\theta_{5}^{\theta_{4}}}{\Gamma(\theta_{4})} \int_{0}^{\infty} a^{\theta_{4} -1}\exp\left(-\left(\frac{d}{\theta_{1}} + \frac{u^{\theta_{3}}}{\theta_{2}} + \theta_{5} \right)a \right) da \\
	&= \frac{\theta_{5}^{\theta_{4}}}{\Gamma(\theta_{4})} \frac{\Gamma(\theta_{4})}{
		\left(\frac{d}{\theta_{1}} + \frac{u^{\theta_{3}}}{\theta_{2}} + \theta_{5}\right)^{\theta_{4}}} \\
	&= \left(\frac{1}{\frac{d}{\theta_{1} \theta_{5}} + \frac{u^{\theta_{3}}}{\theta_{2}\theta_{5}} + 1} \right)^{\theta_{4}}, \quad \theta_{1}, \theta_{2}, \theta_{4}, \theta_{5} > 0, \theta_{3} \in (0, 2].
\end{align*}
After reparameterization, we have that 
\begin{equation}~\label{M5}
	C(d; u) = 	\left(\frac{1}{\frac{d}{\theta_{1}} + \frac{u^{\theta_{3}}}{\theta_{2}} + 1} \right)^{\theta_{4}},  
\end{equation} 
where $\theta_{1}$, $\theta_{2}, \theta_{4} > 0$, and $\theta_{3} \in (0, 2]$ is a valid space-time covariance function on Euclidean trees. The model given by~(\ref{M5}) extends the metric model in Section~\ref{sec:3.2} since it is a function of  $\frac{d}{\theta_{1}} + \frac{u^{\theta_{3}}}{\theta_{2}}$. The covariance function is continuous at the origin and monotonically decays to 0 as $d \rightarrow \infty$ and/or $u \rightarrow \infty$. The spatial marginal is convex on $(0, \infty)$, but this is not necessarily so for the temporal marginal.  

\section{Stream Temperature Example}\label{sec:5}
As mentioned in Section~\ref{sec:2.1} and~\ref{sec:4.1}, a tree-like stream network is naturally a Euclidean tree with the geodesic distance. In this section, we illustrate the use of models introduced in previous sections on a real stream temperature data set. Streams and rivers are important to humans as well as certain plants and animals~\citep{ver2010}. Stream temperature controls many physicochemical processes~\citep{isaak2018} and is one of the key variables affecting habitat suitability for numerous aquatic species~\citep{gallice2016}. Temperature-sensitive species and life stages are most vulnerable during the warm summer season~\citep{isaak2018}. Thus, gaining insights on potentially spatiotemporally varying stream temperatures over networks during the warmest time of the year is of special interest. 

\begin{table}[ht]
	\caption{Summary of geographical variables and mean August water temperature in the year 2014 of 96 survey sites.}\label{tb:geo}
		\centering
		\begin{tabular}{ccccc}
			\hline
			\hline
			& Latitude  & Longitude & Upstream Distance & Mean August Water Temp. \\ 
			\hline
			Min & 45.703 $^\circ$N &  114.493 $^\circ$W & 71.462 km & 7.510 $^\circ$C  \\ 
			Mean & 46.600 $^\circ$N & 115.322 $^\circ$W & 218.348 km & 12.364 $^\circ$C  \\ 
			Max & 46.968 $^\circ$N & 116.530 $^\circ$W & 296.040 km & 21.760 $^\circ$C  \\ 
			\hline
		\end{tabular}
\end{table}

The advent of inexpensive sensors provides a large amount of accessible water temperature records~\citep{isaak2018}. The original data~\citep{isaak2018}, available at \url{https://www.researchgate.net/publication/325933910_Principal_components_of_thermal_regimes_in_mountain_river_networks}, comprises mean daily water temperature (in Celsius) during a 5-year time period (2011 - 2015) at 226 observational sites residing on several adjacent river networks in the northwestern United States. As the proposed models only apply to observations on the same network with finite geodesic distances in between, we consider 96 survey sites from the Clearwater River Basin (Figure~\ref{fig:net}) in central Idaho. Observations are more crowded in the northeast region and geographical summary statistics are given in Table~\ref{tb:geo}. The average rate of missingness of the original water temperature data is about $12\%$~\citep{isaak2018}. Considering computational feasibility, we choose to analyze data from 10 consecutive days 08/01/2014 - 08/10/2014 from the warmest season when records of those 96 sites are complete, which yields a total of 960 observations. Figure~\ref{fig:ste} clearly indicates that spatial and temporal dependence may exist among observations.  

\begin{figure}[H]
	\begin{center}
		\includegraphics[scale=0.6]{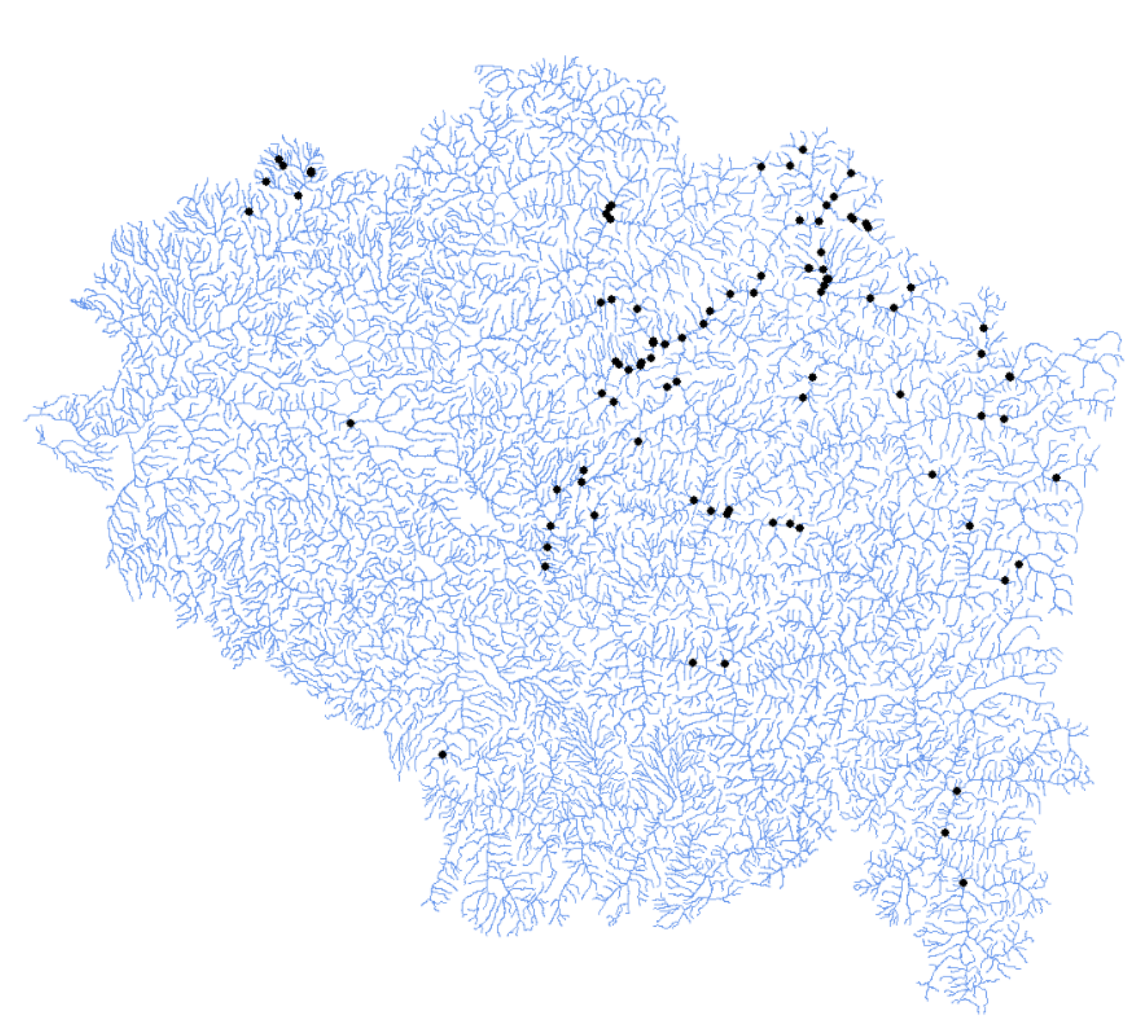}
	\end{center}
	\caption{Delineated stream lines of the Clearwater River Basin with 96 observation sites superimposed on top.}\label{fig:net}
\end{figure}

Stream lines of the study region were downloaded from the National Stream Internet (NSI) dataset (\url{https://www.fs.fed.us/rm/boise/AWAE/projects/NationalStreamInternet/NSI_network.html}). Observations were integrated with the network object using the STARS (Version 2.0.7) toolset in ArcGIS (Version 10.7.1). A Spatial Stream Network (.ssn) object was then created, from which we extracted the topological structure of the network using the SSN package (Version 1.1.1) in R. We consider a regression type of analysis including five available covariates: mean August stream temperature ($^\circ$C), drainage area (km$^2$), reach slope (m/m), elevation (m) and mean annual flow (m$^3$/s). Since a few observations have much larger drainage area and mean annual flow than others, we log transform those two positive variables to make them more symmetrically distributed. We consider each of the covariance functions given by Models 1 - 5, plus a separable model obtained by enforcing the space-time interaction parameter $\beta = 0$ in Model 1, and a common linear regression model with i.i.d. errors. The data is unbalanced in the sense that more than $90\%$ of pairs of observation sites are flow-unconnected. We cannot obtain the exact number of leaves where points lie due to the limitation of current software. However, the parameter space of $\delta$ in Model 3 depends on the number of leaves, and large $\delta$ will result in a rather small covariance value, which makes the estimation of covariance parameters less reliable. Therefore, we let $\delta = 97$ in the following analysis. For Model 4, the positive additive function $\Omega(x)$ is generated using the accumulated drainage area which comes along with the NSI dataset using STARS. 

\begin{figure}[H]
	\begin{center}
		\includegraphics[scale=0.9]{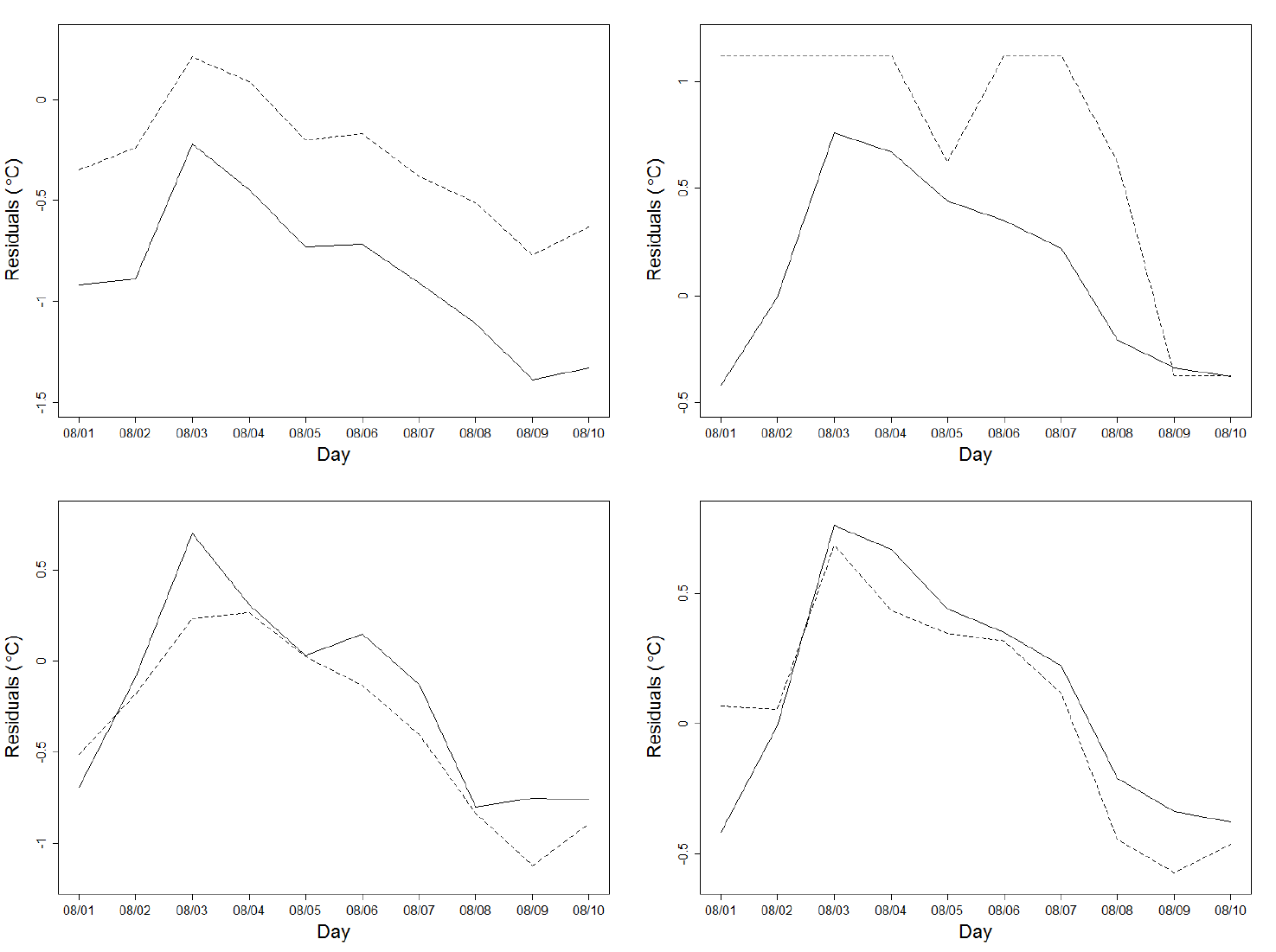}
	\end{center}
	\caption{Time series of residuals from the linear regression model with i.i.d. errors on 08/01/2014 - 08/10/2014 between various pairs of survey sites. Upper left: flow-connected sites with the smallest geodesic distance, i.e. 0.104 km; upper right: flow-connected sites with the largest geodesic distance, i.e. 224.577 km; bottom left:flow-unconnected sites with the smallest geodesic distance, i.e. 0.031 km; bottom right: flow-unconnected sites with the largest geodesic distance, i.e. 453.438 km.}\label{fig:ste}
\end{figure}

Models introduced thus far are continuous with variance $\sigma_{0}^{2} = C_{0}(\bm{s}, \bm{s}; t, t) = 1,\ (\bm{s}, t) \in \mathcal{T} \times \mathbb{R}$. Based on applications in the literature (\citet{gneiting2006} and \citet{gneiting2002}), as well as empirical analysis of space and time marginal functions for the stream temperature data, we allow a discontinuity of the covariance function at the origin by multiplying $C_{0}$ by a constant and adding to it a pure spatial nugget effect $\theta_{\text{nug}} > 0$:
\[
C(\bm{s}_{i}, \bm{s}_{j}; t_{i}, t_{j}) = \sigma^{2}C_{0}(\bm{s}_{i}, \bm{s}_{j}; t_{i}, t_{j}) + \theta_{\text{nug}}\bm{1}_{\bm{s}_i = \bm{s}_j}, \qquad (\bm{s}_{i}, t_{i}) \in \mathcal{T} \times \mathbb{R}, \ i= 1, 2.
\]  
The spatial nugget effect accounts for measurement error and micro-scale spatial variability~\citep[pp.~58]{cressie1993}. The exploratory analysis of fitted residuals suggests that the normal distribution assumption is reasonable. 

Our primary goal of this data analysis is prediction. In order to evaluate different models' predictive performances, we conduct an 8-fold cross validation. We randomly split 96 sites into 8 non-overlapping subsamples of equal size. For each cross validation replication, we fit 7 models on $84 \times 10 = 840$ data points by maximum likelihood estimation using the \texttt{optim} function in R and predict daily stream temperature by the best linear unbiased (universal kriging) predictor for the hold-out sample on each of 10 days. We assess performance using in-sample information based criteria, log likelihood score and Bayesian Information Criterion (BIC), and out-of-sample root mean squared prediction error (RMSPE) and continuous ranked probability score (CRPS; see details in~\citet{gneiting2006}). Results are summarized in Table~\ref{tb:cv}
. 

\begin{table}[ht]
	\centering
	\caption{Comparison of log likelihood values (LL), Bayesian Information Criteria (BIC), root mean squared prediction errors (RMSPE) and continuous ranked probability scores (CRPS) averaged over 8-fold cross validation replications for Models 1 - 5, the separable model and the linear regression (LR) model on the water temperature data with covariates introduced above.}\label{tb:cv}
		\begin{tabular}{cccccccc}
			\hline 
			\hline
			& Model 1 & Model 2 & Model 3 & Model 4 & Model 5 & Separable & LR \\ 
			\hline
			LL & -71.753 & -131.864 & -185.333 & -491.083 & -260.041 & -71.859 & -1291.265\\ 
			BIC & 237.774 & 351.263 & 444.733 & 1062.966 & 600.883 & 231.252 & 2629.664\\ 
			RMSPE & 1.142 & 0.938 & 1.069 & 0.999 & 0.949 & 1.141 & 1.124\\ 
			CRPS & 0.637 & 0.502 & 0.600 & 0.536 & 0.526 & 0.636 & 0.610\\ 
			\hline
		\end{tabular}
\end{table}

Large log likelihood value and small BIC indicate the model fits the training data well, while small RMSPE and CRPS suggest little discrepancy between the predicted and observed values from the test data. Table~\ref{tb:cv} reveals that in-sample and out-of-sample model evaluation metrics favor different models. Though Model 1 and the corresponding separable model outperform the rest in terms of log likelihood score and BIC, they have the worst and almost identical prediction precision based on RMSPE and CRPS. Similar model performances do not necessarily suggest that $\beta$ should be zero in Model 1. In fact, the space-time interaction is quite strong as the cross validation estimate of $\beta$ sticks at the upper bound, i.e. 1. Surprisingly, even the linear regression model predicts the hold-out samples more precisely than Model 1. Note that Model 1 also has the most covariance parameters. We interpret the somewhat contradictory results that a model behaves the best in-sample but worst out-of-sample as there might exist over-parameterization and model misspecification issues. Since our focus is the predictive performance, we prefer Model 2, which is constructed based on the Hilbert space embedding technique as well, because it has the smallest cross validation RMSPE and CRPS. It reduces the former metric by 17.86\% and 16.55\%, and latter by 21.19\% and 17.70\%, compared to Model 1 (also the separable model) and the linear regression model, respectively. Besides, its in-sample performance is only inferior to Model 1 and the separable model.  Thus, Model 2 is the overall ``best" model for this data set and fitted marginal functions are given in Figure~\ref{fig:fit}. 

\begin{figure}[H]
	\begin{center}
		\includegraphics[scale=1]{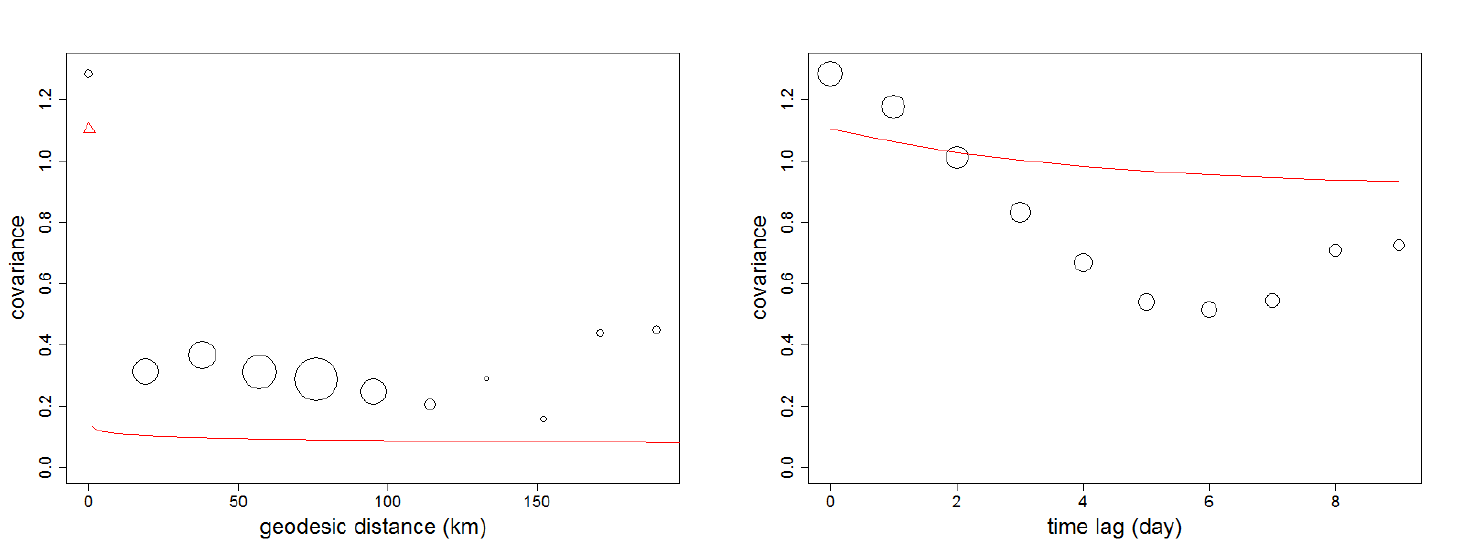}
	\end{center}
	\caption{Empirical (points) and fitted (red lines) marginal covariance functions of residuals from Model 2 by maximum likelihood estimation on the full data. Open circles are sample covariograms with size proportional to number of data pairs within each group. }\label{fig:fit}
\end{figure}

Model 3 is the most parsimonious space-time model we have applied and saves considerable computation time. For instance, it took about 50 minutes to complete 8 replications of maximum likelihood estimation for Model 3 on the Linux kernel 5.7.14, which is about one-third the computation time of Model 1 and Model 2. With comparable predictive performance, Model 3 has an overwhelming computational advantage over the others, hence should be given special attention. Though Model 4 considers flow direction, it does not help much with the model performance (both in-sample and out-of-sample) for this specific data set. A few explanations are: the  tail-up/tail-down, as well as the time series components are not the most appropriate ones; weights that are calculated based on the accumulated drainage area are not reasonable; some observation sites are distributed quite far away from the rest, which might hamper the effectiveness of the analysis, etc. However, selecting the best model for this specific data set is beyond the scope of this analysis, and we leave it for future research.  

Another interesting issue is how these models' predictive performances compare to models in the literature. Direct comparison is unavailable since there do not exist space-time covariance models on the network and researches on this specific data set are limited. But~\citet{gallice2016} did an extensive literature review on statistical stream temperature prediction in ungauged basins on different time scales. Despite different study regions and models applied, the daily RMSPEs range from 1.0 - 2.7$^\circ$C (Table 2.1,~\citet{gallice2016}). Moreover, the measurement errors from different sensors in the original data are from $\pm 0.2$ to $\pm 0.5 ^\circ$C~\citep{isaak2018}. In light of this, all of our proposed models achieve decent prediction precision.  

\section{Discussion}\label{sec:6}
This article presented a collection of tools to build valid space-time covariance models on generalized linear networks, and on an important subclass, Euclidean trees. We studied examples obtained by each constructive method and applied them to a real stream temperature data. We have not yet given a standard guidance on how to choose suitable candidate models for an arbitrary data set. But understanding the underlying physical process~\citep{gneiting2002} and matching geometric features of theoretical covariance functions to the empirical space-time covariance surface~\citep{de2013} would be helpful.

It has been argued that when prediction is the goal, model estimation is just a means to an end~\citep{porcu2019}. We notice that maximum likelihood estimates of $\sigma^{2}$ for Model 3 and 5 are much higher, i.e. $\hat{\sigma}^{2} = 15.426$ and $13.137$, respectively, than the rest. It does not cause much trouble in this particular analysis as both models have reasonable prediction precision, as determined by cross validation. However, the study of parameter estimability under infill and increasing domain asymptotics on a network is an open question that needs to be addressed and requires special attention. For instance, there is a wide range of geodesic distances between pairs of observation sites, i.e. from less than 0.1 km to 453.438 km, in the water temperature data. Precisely characterizing the dynamics of the space-time process can help refine future sampling designs and reduce data redundancy~\citep{isaak2018}.  

Though we emphasized the decisive role of valid covariance functions in geostatistical models, they also allow direct extension to space-time log Gaussian Cox processes on generalized networks (\citet{moller1998}; \citet{anderes2020}), which we leave for future investigation. 

\bigskip
\begin{center}
	{\Large\bf Acknowledgments}
\end{center}
This article is a portion of Jun Tang's Ph.D. thesis.

\bigskip
\begin{center}
{\Large\bf Supplementary Material}
\end{center}

\begin{description}

\item[Supplement A:] Proofs of Lemma 3 and Theorem~\ref{thm5}.

\item[Supplement B:] Visualizations of marginal functions, along with the covariance surface of models developed from each class.

\end{description}

\bibliographystyle{Chicago}
\bibliography{ref}

\section*{Supplement A: Proofs of Lemma 3 and Theorem~\ref{thm5}}
\begin{customlemma}{3}\label{lemma3}
	Suppose that $\mathcal{T}$ is a Euclidean tree with $m$ leaves, where $m \geq 3$. Let $\psi(t) = t^{\lambda} + \beta$ for $t \geq 0$, where $0 < \lambda \leq 1$ and $\beta > 0$. Then $\psi \in CND(\mathcal{T}, d_{\cdot, \mathcal{T}})$. 
\end{customlemma}

\begin{proof}
	As in Example 2.1 of~\citet{zastavnyi2011}, let $\psi_{0}: \mathbb{R}^{n} \rightarrow \mathbb{R}$ be such that $\psi_{0}(\bm{x}) =  ||\bm{x}||_{1}^{\lambda} + \beta$ where $ \beta > 0$ and $0 < \lambda \leq 1$. Then $\exp(-c \psi_{0}(\bm{x}))$ is positive definite for all positive $c$, given $n \geq 2$. It follows by Schoenberg's (Example 2.4 in~\citet{zastavnyi2011}) that 
	\[
	\sum_{i=1}^{N}\sum_{j =1}^{N}a_{i} a_{j}\psi_{0}(\bm{x}_{i} - \bm{x}_{j}) \leq 0,
	\]
	for any finite collection $\{a_{i}\}_{i=1}^{N} \subset \mathbb{R}$ and $\{\bm{x}_{i}\}_{i=1}^{N} \subset \mathbb{R}^{n}$, given that $\sum_{i=1}^{N}a_{i} = 0$. 
	
	In concert with Theorem 3, there exists a mapping $i: \mathcal{T} \rightarrow \mathbb{R}^{\ceil[\big]{\frac{m}{2}}}$ such that for every finite collection $\bm{s}_{1}, \cdots, \bm{s}_{n} \in \mathcal{T}$ and $a_{1}, \cdots, a_{N} \in \mathbb{R}$ with $\sum_{i = 1}^{N} a_{i} = 0$, 
	\[
	\sum_{i=1}^{N}\sum_{j = 1}^{N}a_{i}a_{j}\psi(d_{\cdot, \mathcal{T}}(\bm{s}_{i}, \bm{s}_{j})) = \sum_{i=1}^{N}\sum_{j = 1}^{N}a_{i}a_{j}\psi(\rho_{1}(i(\bm{s}_{i}) - i(\bm{s}_{j}))) = \sum_{i=1}^{N}\sum_{j = 1}^{N}a_{i}a_{j}\psi_{0}(i(\bm{s}_{i}) - i(\bm{s}_{j})) \leq 0.  
	\]
\end{proof}

\vspace{4mm}

\begin{proof}[Proof of Theorem 5]
	The proof of sufficiency is trivial so we focus on the necessity. The tail-down covariance functions for flow-connected and flow-unconnected sites are given as follows:
	\begin{equation}\label{fc_cov} 
	C_{c}(d) = \int_{-\infty}^{-d}g(-x)g(-x - d) dx 
	\end{equation}
	and 
	\begin{equation}\label{fu_cov}
	C_{n}(a, b) = \int_{-\infty}^{-b}g(-x)g(-x - (b - a)) dx = \int_{-\infty}^{-b}g(-x)g(-x - d + 2a) dx,
	\end{equation}
	where $d, a, b$ are given in Section 4.1, $d = a + b > 0$, $b \geq a \geq 0$, and $g(-x)$ is a unilateral tail-down kernel with nonzero values only on the negative side of $0$. 
	
	Letting $y \equiv x -a$, (\ref{fu_cov}) can be re-written as:
	\begin{equation}\label{fu_cov2}
	C_{n}(a, b) =  \int_{-\infty}^{-(a + b)} g(-y - a)g (-y - b) dy. 
	\end{equation}
	
	Assume that $g(x)$ is a continuous function on the positive half of the real line. Let $f(a, b) := \int_{-\infty}^{-(a + b)} \left\{g(-x)g(-x - (a + b)) - g(-x -a)g(-x - b) \right\} dx$, which is the difference between~(\ref{fc_cov}) and~(\ref{fu_cov2}). If the covariance model is isotropic, i.e. if~(\ref{fc_cov}) =~(\ref{fu_cov2}), $\forall b\geq a\geq 0,\ d = a + b > 0$, it follows that $f(a, b) = 0$. By the Leibniz integral rule, we take the partial derivative with respective to $a$ of both sides: 
	\begin{align}\label{par_a}
	& \frac{\partial f(a, b)}{\partial a}  = 0\nonumber \\
	\Longrightarrow   -\left[g(a + b)g(0) - g(b)g(a) \right] + & \int_{-\infty}^{-(a + b)} \left\{-g(-x)g'(-x - (a + b)) + g'(-x - a) g(-x - b) \right\} dx  = 0 \nonumber \\
	\Longrightarrow  g(0)g(a + b) - g(a) g(b) + & \int_{-\infty}^{-(a + b)} \left\{g(-x)g'(-x-(a + b)) - g'(-x -a)g(-x - b) \right\} dx = 0.
	\end{align}
	
	Let $ y = -x - (a + b)$. By the change of variable formula, the second half on the left hand side of~(\ref{par_a}) is equivalent to
	\begin{align}\label{par_a_2}
	\int_{-\infty}^{-(a + b)} \left\{g(-x)g'(-x-(a + b)) - g'(-x -a)g(-x - b)\right\}dx = \nonumber \\ \int_{0}^{\infty} \left\{ g(y + a)g'(y+ b) - g'(y)g(y + (a + b))  \right\}dy. 
	\end{align}
	Similarly, by the change of variable, $f(a, b)$ can be re-written as $f^*(a, b)$, such that
	\[
	f^{*}(a, b) = \int_{0}^{\infty}\left\{ g(y)g(y + (a + b)) - g(y + a)g(y + b) \right\}dy = 0.
	\]
	Taking the partial derivative with respect to $b$ yields
	\begin{equation}\label{par_b}
	\frac{\partial f^{*}(a, b)}{\partial b} = \int_{0}^{\infty} \left\{g(y)g'(y + (a + b)) - g(y + a)g'(y + b) \right\}dy = 0. 
	\end{equation}
	Plugging ~(\ref{par_b}) back to~(\ref{par_a_2}) gives
	\begin{align}\label{par_b_2}
	\int_{-\infty}^{-(a + b)} \left\{g(-x)g'(-x-(a + b)) - g'(-x -a)g(-x - b)\right\}dx = \nonumber \\ \int_{0}^{\infty} \left\{ g(y)g'(y + (a + b)) - g'(y)g(y + (a + b)) \right\} dy.
	\end{align}
	Notice that when $b = 0$,~(\ref{par_b}) becomes
	\[
	\int_{0}^{\infty} \left\{g(y)g'(y + a) - g'(y)g(y + a) \right\} dy = 0, \quad \forall a \geq 0, 
	\]
	which implies~(\ref{par_b_2}) = 0 as well. Combined with~(\ref{par_a}) and~(\ref{par_a_2}), it follows that 
	\[
	g(0)g(a + b) - g(a)g(b) = 0, \quad \forall b \geq a \geq 0.
	\]
	
	Now, take logarithm of both sides 
	\begin{align}\label{const}
	\log g(a + b) + \log g(0) &= \log g(a) + \log g(b) \nonumber \\
	\Longrightarrow \underset{b \downarrow 0}{\lim} \frac{\log g(a + b) - \log g(a)}{b} &= \underset{b \downarrow 0}{\lim} \frac{\log g(b) - \log g(0)}{b} \nonumber  \\
	\Longrightarrow \left[\log g(x) \right]' \Big|_{x = a} &= \left[\log g(x) \right]'\Big|_{x = 0^{+}}. 
	\end{align}
	(\ref{const}) implies that $\left[\log g(x) \right]'$ is a constant, say $c_{1}$, when $x > 0$. Thus,
	\begin{equation}\label{final}
	\left[\log g(x) \right]' = c_{1} \Longleftrightarrow \log g(x)  = c_{1}x + c_{2} \Longleftrightarrow g(x) = \exp (c_{1}x + c_{2}). 
	\end{equation}
	In order to make improper integrals~(\ref{fc_cov}) and~(\ref{fu_cov}) convergent, $c_{1}$ must be negative. Then,~(\ref{final}) is exactly the same as the exponential kernel. 
\end{proof}

\newpage
\section*{Supplement B: Plots of Marginal Functions and 3-D Surfaces}
\begin{figure}[H]
	\begin{center}
		\includegraphics[scale=1]{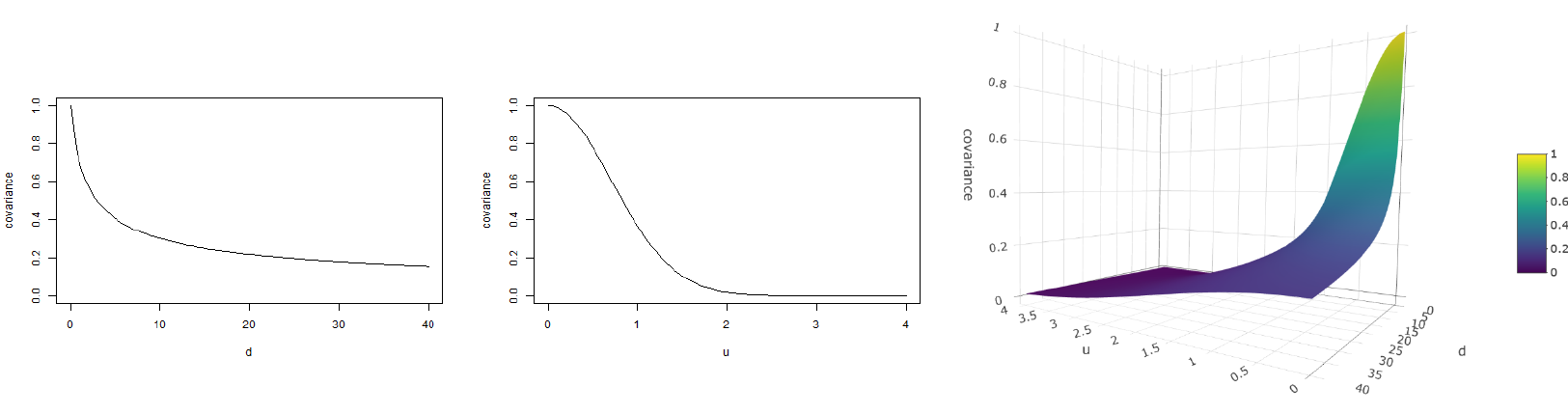}
		\caption{Marginal plots and 3-D surface for Model 1 with covariance parameters $(c, \nu, k, \beta, \tau, b) = (1, 1, 1, 0.5, 0.5, 1)$. }
	\end{center}
\end{figure}

\begin{figure}[H]
	\begin{center}
		\includegraphics[scale=1]{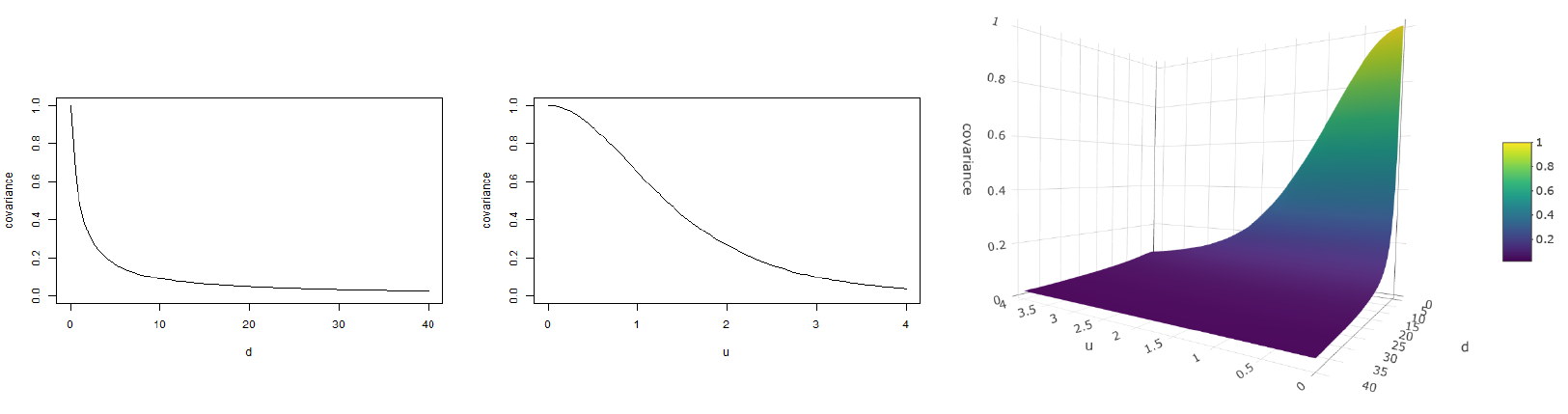}
		\caption{Marginal plots and 3-D surface for Model 2 with covariance parameters $(a, \alpha, b, c, \nu) = (1, 1, 1, 1, 1)$. }
	\end{center}
\end{figure}

\begin{figure}[ht]
	\begin{center}
		\includegraphics[scale=1]{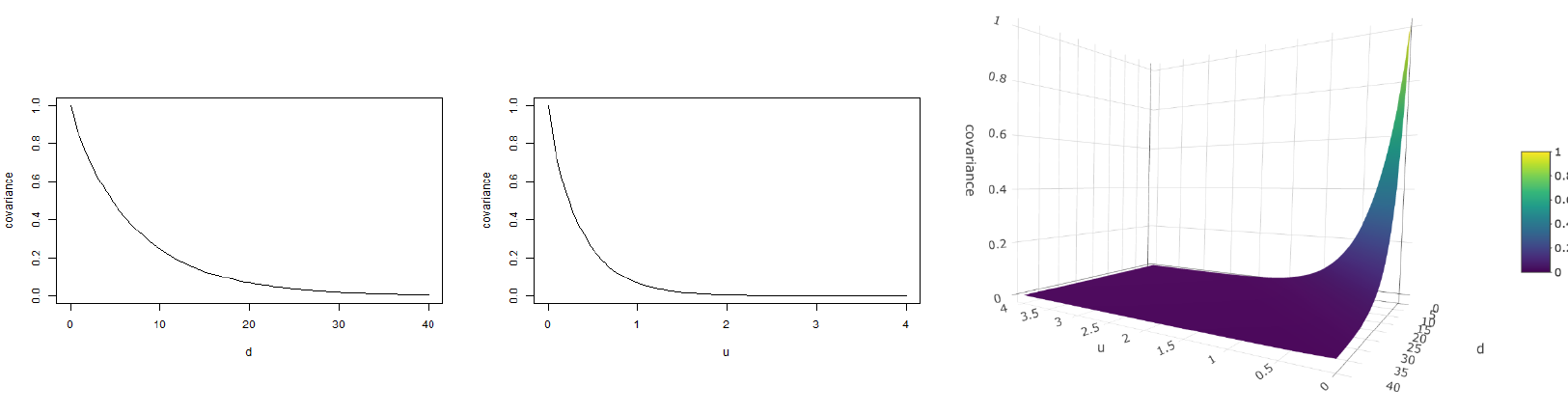}
		\caption{Marginal plots and 3-D surface for Model 3 with covariance parameters $(\alpha, \beta, \nu, \delta) = (200, 10, 0.9, 20)$. }
	\end{center}
\end{figure}

\begin{figure}[ht]
	\begin{center}
		\includegraphics[scale=1]{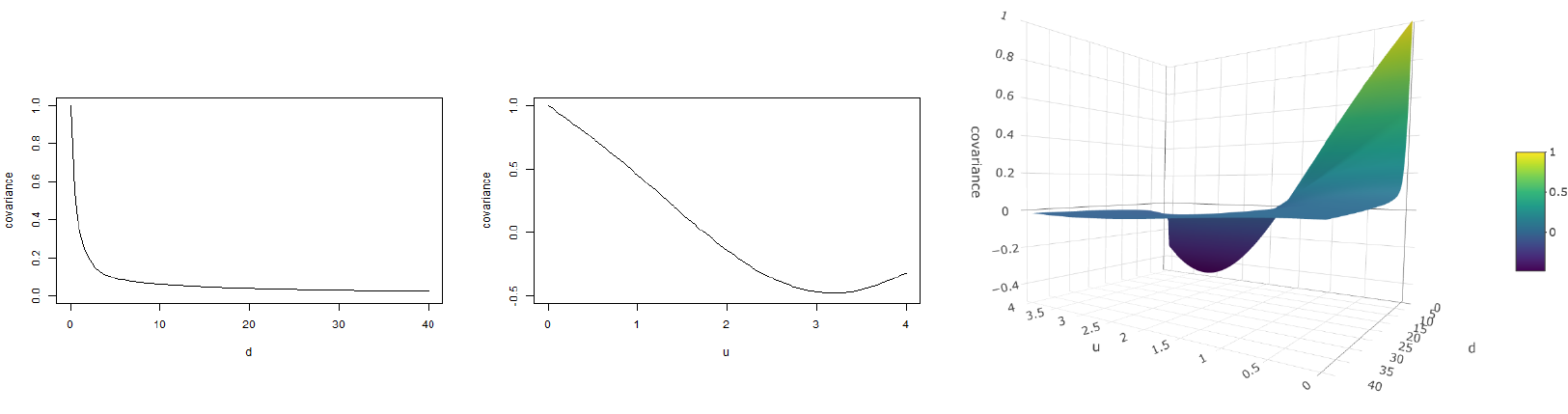}
		\caption{Marginal plots and 3-D surface for Model 4 between flow-connected sites with weights fixed at $0.5$ and  covariance parameters $(\theta_{1}, \theta_{2}, \theta_{3}, \theta_{4}) = (1, 1, 1, 1)$. }
	\end{center}
\end{figure}

\begin{figure}[ht]
	\begin{center}
		\includegraphics[scale=1]{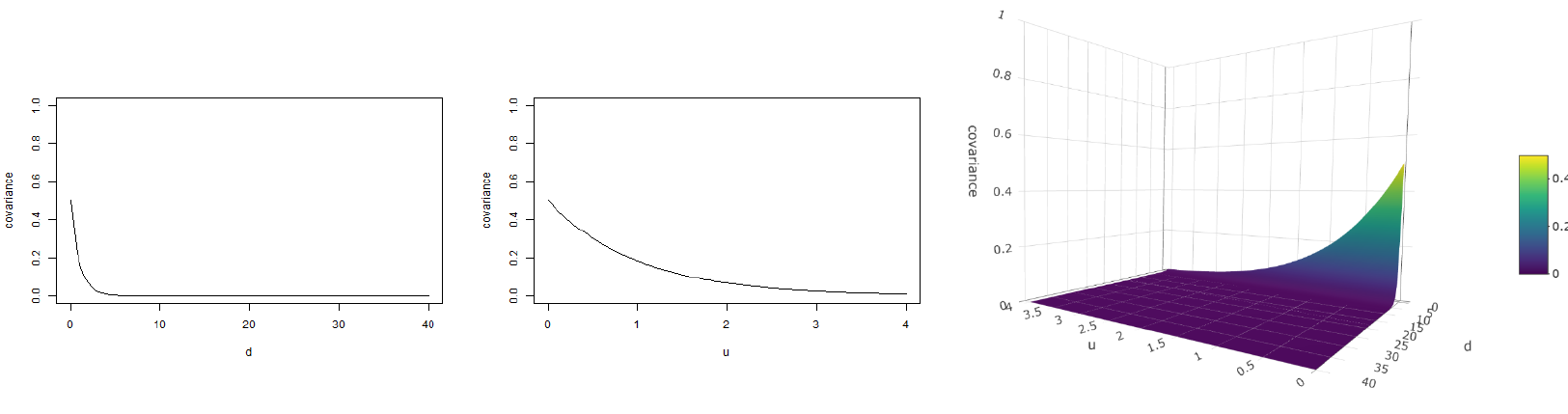}
		\caption{Marginal plots and 3-D surface for Model 4 between flow-unconnected sites with covariance parameters $(\theta_{1}, \theta_{2}, \theta_{3}, \theta_{4}) = (1, 1, 1, 1)$. }
	\end{center}
\end{figure}

\begin{figure}[ht]
	\begin{center}
		\includegraphics[scale=1]{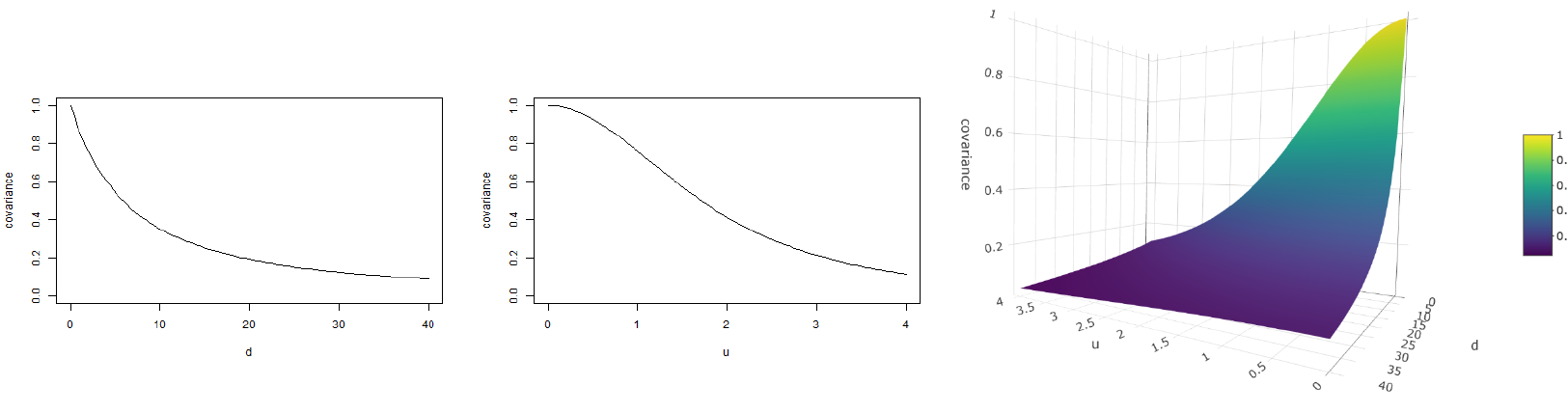}
		\caption{Marginal plots and 3-D surface for Model 5 with covariance parameters $(\theta_{1}, \theta_{2}, \theta_{3}, \theta_{4}) = (10, 5, 1.5, 1)$. }
	\end{center}
\end{figure}

\begin{figure}[ht]
	\begin{center}
		\includegraphics[scale=1]{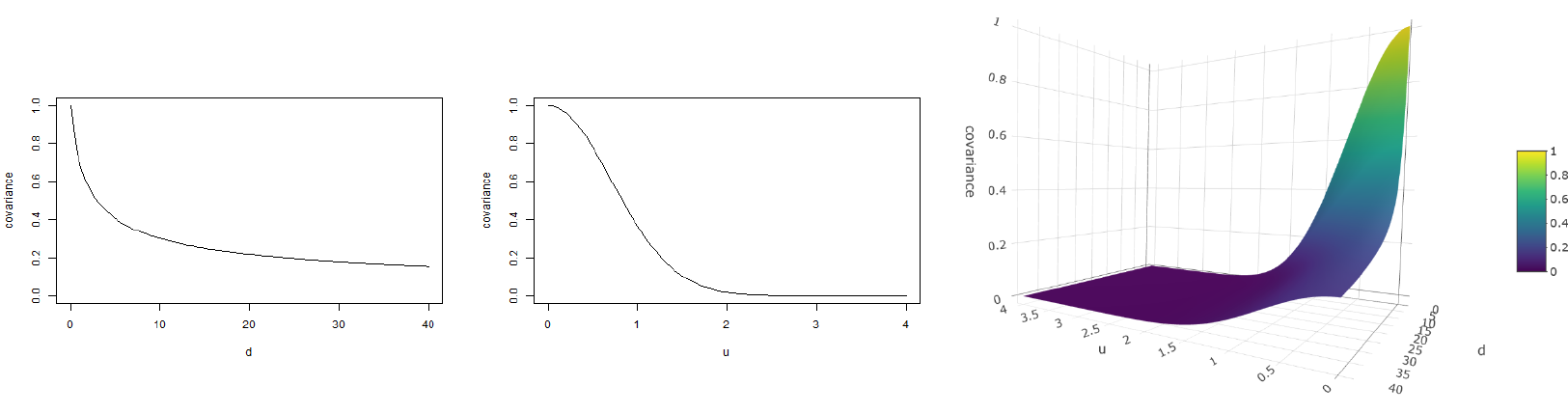}
		\caption{Marginal plots and 3-D surface for the separable with covariance parameters $(c, \nu, k, \tau, b) = (1, 1, 1, 0.5, 1)$. }
	\end{center}
\end{figure}

\end{document}